\documentclass[aps,pra,letterpaper,nofootinbib,10pt
,superscriptaddress
,notitlepage,twocolumn,eqsecnum,showpacs]{revtex4-1}
\usepackage{color}
\usepackage{amsmath,amsfonts,amssymb,amsthm}
\usepackage{graphicx}
\PassOptionsToPackage{caption=false}{subfig}
\usepackage{subfig}
\usepackage{enumerate}

\usepackage{tikz}
\usetikzlibrary{arrows,decorations.pathmorphing,backgrounds,positioning,fit}
\usepackage{epstopdf} 
\usepackage{bm}  

\usepackage[pdfpagelabels,pdftex,bookmarks,breaklinks]{hyperref}
\definecolor{darkblue}{RGB}{0,0,127} 
\definecolor{darkgreen}{RGB}{0,150,0}
\hypersetup{colorlinks, linkcolor=darkblue, citecolor=darkgreen, filecolor=red, urlcolor=blue}

\pgfrealjobname{OpticalGaussian_stripped}

\newtheorem{theorem}{Theorem}

\newtheorem{lemma}[theorem]{Lemma}

\newtheorem{definition}[theorem]{Definition}

\renewenvironment{proof}[1][Proof]{\noindent\textbf{#1.} }{\ $\Box$}

\def\Z{\mathbb{Z}}

\newcommand{\ket}[1]{|{#1}\rangle}

\input{Qcircuit}
\usepackage{mathtools}

\usepackage{stmaryrd}
\usepackage{color}
\usepackage{ifthen}
\usepackage{amsthm}
\usetikzlibrary{chains,calc,fit,shapes.geometric,decorations.pathmorphing}

\def\micronodesize{4pt}
\def\edgethickness{very thick}
\def\poscolor{blue!60!black}
\def\negcolor{red!30!yellow}

\def\unitcolor{\poscolor}
\def\evencolor{black}
\def\oddcolor{white}
\def\phantomcolor{gray!90}
\def\phantomfade{very nearly transparent}

\newcommand{\reals}{\mathbb{R}}

\newcommand\tp{\mathrm{T}}
\newcommand\herm{\mathrm{H}}

\newcommand\group[1]{\mathrm{#1}}
\newcommand\algebra[1]{\mathfrak{#1}}

\tikzset{>=to}

\tikzset{micro-no-color/.style={
	circle,
	minimum size=\micronodesize,
	inner sep=0pt
	}
}

\tikzset{micro/.style={
	micro-no-color, ball color=black,
	}
}

\tikzset{micro-even/.style={
	micro-no-color, ball color=\evencolor,
	}
}

\tikzset{micro-odd/.style={
	micro-no-color, ball color=\oddcolor,
	}
}

\tikzset{poslink/.style={
	\edgethickness,
	draw=\poscolor,
	nearly opaque
	}
}

\tikzset{neglink/.style={
	\edgethickness,
	draw=\negcolor,
	nearly opaque
	}
}

\tikzset{unitlink/.style={
	\edgethickness,
	draw=\unitcolor,
	nearly opaque
	}
}

\tikzset{dashlink/.style={
	\edgethickness,
	draw=\unitcolor,
	dash pattern=on 0.75pt off 0.75pt,
	nearly opaque
	}
}

\tikzset{fulllink/.style={
	\edgethickness,
	draw=\unitcolor,
	opaque
	}
}

\tikzset{phantomlink/.style={
	\edgethickness,
	draw=\phantomcolor,
	\phantomfade
	}
}

\tikzset{optpath/.style={
	thick,
	draw=black
	}
}

\tikzset{optdelay/.style={
	optpath,
	decorate,
	decoration={coil,segment length=4pt,pre=lineto,pre length=1mm,post length=1mm}
	}
}

\tikzset{optdelaylong/.style={
	optpath,
	decorate,
	decoration={coil,segment length=2pt, pre=lineto, pre length=1mm,post length=1mm}
	}
}

\tikzset{nodehighlight/.style={
	semithick,
	red,
	fill=red,
	semitransparent
	}
}

\def\latticeoptsscale{0.8}

\tikzset{latticeopts/.style={
	baseline=-\latticeoptsscale*1.5cm-1ex,
	x= \latticeoptsscale*3.375cm,
	y= \latticeoptsscale*3.75mm,
	z= \latticeoptsscale*20.25mm,
	inner sep=0pt,
	outer sep=0pt
	}
}

\newcommand\piplus[4] %
{
	\draw [poslink] (#1) -- (#3);
	\draw [poslink] (#1) -- (#4);
	\draw [poslink] (#2) -- (#3);
	\draw [poslink] (#2) -- (#4);	
}

\newcommand\piminus[4] %
{
	\draw [poslink] (#1) -- (#3);
	\draw [neglink] (#1) -- (#4);
	\draw [neglink] (#2) -- (#3);
	\draw [poslink] (#2) -- (#4);	
}

\newcommand\beamsplit[4] %
{
	\draw [neglink] (#1) -- (#3);
	\draw [neglink] (#1) -- (#4);
	\draw [poslink] (#2) -- (#3);
	\draw [poslink] (#2) -- (#4);	
}

\newcommand\tmslink[4] %
{
	\draw [phantomlink] (#1) -- (#3);	
	\draw [phantomlink] (#1) -- (#4);	
	\draw [unitlink] (#2) -- (#3);	
	\draw [phantomlink] (#2) -- (#4);	
}

\newcommand\tmslinky[4] %
{
	\draw [unitlink] (#2) -- (#3);	
}

\newcommand\linearlink[4] %
{
	\draw [phantomlink] (#1) -- (#3);	
	\draw [phantomlink] (#1) -- (#4);	
	\draw [phantomlink] (#2) -- (#3);	
	\draw [unitlink] (#2) -- (#4);	
}

\newcommand\Pizero[8] %
{
	\draw [poslink] (#1) -- (#5);
	\draw [poslink] (#1) -- (#6);
	\draw [poslink] (#1) -- (#7);
	\draw [poslink] (#1) -- (#8);
	\draw [poslink] (#2) -- (#5);
	\draw [poslink] (#2) -- (#6);	
	\draw [poslink] (#2) -- (#7);
	\draw [poslink] (#2) -- (#8);	
	\draw [poslink] (#3) -- (#5);
	\draw [poslink] (#3) -- (#6);
	\draw [poslink] (#3) -- (#7);
	\draw [poslink] (#3) -- (#8);
	\draw [poslink] (#4) -- (#5);
	\draw [poslink] (#4) -- (#6);	
	\draw [poslink] (#4) -- (#7);
	\draw [poslink] (#4) -- (#8);	
}

\newcommand\Pione[8] %
{
	\draw [poslink] (#1) -- (#5);
	\draw [neglink] (#1) -- (#6);
	\draw [poslink] (#1) -- (#7);
	\draw [neglink] (#1) -- (#8);
	\draw [neglink] (#2) -- (#5);
	\draw [poslink] (#2) -- (#6);	
	\draw [neglink] (#2) -- (#7);
	\draw [poslink] (#2) -- (#8);	
	\draw [poslink] (#3) -- (#5);
	\draw [neglink] (#3) -- (#6);
	\draw [poslink] (#3) -- (#7);
	\draw [neglink] (#3) -- (#8);
	\draw [neglink] (#4) -- (#5);
	\draw [poslink] (#4) -- (#6);	
	\draw [neglink] (#4) -- (#7);
	\draw [poslink] (#4) -- (#8);	
}

\newcommand\Pitwo[8] %
{
	\draw [poslink] (#1) -- (#5);
	\draw [poslink] (#1) -- (#6);
	\draw [neglink] (#1) -- (#7);
	\draw [neglink] (#1) -- (#8);
	\draw [poslink] (#2) -- (#5);
	\draw [poslink] (#2) -- (#6);	
	\draw [neglink] (#2) -- (#7);
	\draw [neglink] (#2) -- (#8);	
	\draw [neglink] (#3) -- (#5);
	\draw [neglink] (#3) -- (#6);
	\draw [poslink] (#3) -- (#7);
	\draw [poslink] (#3) -- (#8);
	\draw [neglink] (#4) -- (#5);
	\draw [neglink] (#4) -- (#6);	
	\draw [poslink] (#4) -- (#7);
	\draw [poslink] (#4) -- (#8);	
}

\newcommand\Pithree[8] %
{
	\draw [poslink] (#1) -- (#5);
	\draw [neglink] (#1) -- (#6);
	\draw [neglink] (#1) -- (#7);
	\draw [poslink] (#1) -- (#8);
	\draw [neglink] (#2) -- (#5);
	\draw [poslink] (#2) -- (#6);	
	\draw [poslink] (#2) -- (#7);
	\draw [neglink] (#2) -- (#8);	
	\draw [neglink] (#3) -- (#5);
	\draw [poslink] (#3) -- (#6);
	\draw [poslink] (#3) -- (#7);
	\draw [neglink] (#3) -- (#8);
	\draw [poslink] (#4) -- (#5);
	\draw [neglink] (#4) -- (#6);	
	\draw [neglink] (#4) -- (#7);
	\draw [poslink] (#4) -- (#8);	
}

\newcommand\BeamSplitx[8] %
{
	\draw [neglink] (#1) -- (#5);
	\draw [neglink] (#1) -- (#6);
	\draw [neglink] (#1) -- (#7);
	\draw [neglink] (#1) -- (#8);
	\draw [poslink] (#2) -- (#5);
	\draw [poslink] (#2) -- (#6);	
	\draw [poslink] (#2) -- (#7);
	\draw [poslink] (#2) -- (#8);	
	\draw [neglink] (#3) -- (#5);
	\draw [neglink] (#3) -- (#6);
	\draw [neglink] (#3) -- (#7);
	\draw [neglink] (#3) -- (#8);
	\draw [poslink] (#4) -- (#5);
	\draw [poslink] (#4) -- (#6);	
	\draw [poslink] (#4) -- (#7);
	\draw [poslink] (#4) -- (#8);	
}

\newcommand\BeamSplitz[8] %
{
	\draw [neglink] (#1) -- (#5);
	\draw [neglink] (#1) -- (#6);
	\draw [poslink] (#1) -- (#7);
	\draw [poslink] (#1) -- (#8);
	\draw [poslink] (#2) -- (#5);
	\draw [poslink] (#2) -- (#6);	
	\draw [neglink] (#2) -- (#7);
	\draw [neglink] (#2) -- (#8);	
	\draw [poslink] (#3) -- (#5);
	\draw [poslink] (#3) -- (#6);
	\draw [neglink] (#3) -- (#7);
	\draw [neglink] (#3) -- (#8);
	\draw [neglink] (#4) -- (#5);
	\draw [neglink] (#4) -- (#6);	
	\draw [poslink] (#4) -- (#7);
	\draw [poslink] (#4) -- (#8);	
}

\newcommand\tmslinkx[8] %
{
	\draw [phantomlink] (#1) -- (#5);
	\draw [phantomlink] (#1) -- (#6);
	\draw [phantomlink] (#1) -- (#7);
	\draw [phantomlink] (#1) -- (#8);
	\draw [unitlink] (#2) -- (#5);
	\draw [phantomlink] (#2) -- (#6);	
	\draw [phantomlink] (#2) -- (#7);
	\draw [phantomlink] (#2) -- (#8);	
	\draw [phantomlink] (#3) -- (#5);
	\draw [phantomlink] (#3) -- (#6);
	\draw [phantomlink] (#3) -- (#7);
	\draw [phantomlink] (#3) -- (#8);
	\draw [phantomlink] (#4) -- (#5);
	\draw [phantomlink] (#4) -- (#6);	
	\draw [phantomlink] (#4) -- (#7);
	\draw [phantomlink] (#4) -- (#8);
}

\newcommand\tmslinkz[8] %
{
	\draw [phantomlink] (#1) -- (#5);
	\draw [phantomlink] (#1) -- (#6);
	\draw [phantomlink] (#1) -- (#7);
	\draw [phantomlink] (#1) -- (#8);
	\draw [phantomlink] (#2) -- (#5);
	\draw [phantomlink] (#2) -- (#6);	
	\draw [phantomlink] (#2) -- (#7);
	\draw [phantomlink] (#2) -- (#8);	
	\draw [phantomlink] (#3) -- (#5);
	\draw [phantomlink] (#3) -- (#6);
	\draw [phantomlink] (#3) -- (#7);
	\draw [phantomlink] (#3) -- (#8);
	\draw [phantomlink] (#4) -- (#5);
	\draw [phantomlink] (#4) -- (#6);	
	\draw [unitlink] (#4) -- (#7);
	\draw [phantomlink] (#4) -- (#8);
}

\newcommand\linearlinklattice[8] %
{
	\draw [phantomlink] (#1) -- (#5);
	\draw [phantomlink] (#1) -- (#6);
	\draw [phantomlink] (#1) -- (#7);
	\draw [phantomlink] (#1) -- (#8);
	\draw [phantomlink] (#2) -- (#5);
	\draw [phantomlink] (#2) -- (#6);	
	\draw [phantomlink] (#2) -- (#7);
	\draw [phantomlink] (#2) -- (#8);	
	\draw [phantomlink] (#3) -- (#5);
	\draw [phantomlink] (#3) -- (#6);
	\draw [phantomlink] (#3) -- (#7);
	\draw [phantomlink] (#3) -- (#8);
	\draw [phantomlink] (#4) -- (#5);
	\draw [phantomlink] (#4) -- (#6);	
	\draw [phantomlink] (#4) -- (#7);
	\draw [unitlink] (#4) -- (#8);
}

\newcommand\beamsplitx[8] %
{
	\draw [neglink] (#1) -- (#5);
	\draw [neglink] (#1) -- (#6);
	\draw [phantomlink] (#1) -- (#7);
	\draw [phantomlink] (#1) -- (#8);
	\draw [poslink] (#2) -- (#5);
	\draw [poslink] (#2) -- (#6);	
	\draw [phantomlink] (#2) -- (#7);
	\draw [phantomlink] (#2) -- (#8);	
	\draw [phantomlink] (#3) -- (#5);
	\draw [phantomlink] (#3) -- (#6);
	\draw [phantomlink] (#3) -- (#7);
	\draw [phantomlink] (#3) -- (#8);
	\draw [phantomlink] (#4) -- (#5);
	\draw [phantomlink] (#4) -- (#6);	
	\draw [phantomlink] (#4) -- (#7);
	\draw [phantomlink] (#4) -- (#8);
}

\newcommand\beamsplitz[8] %
{
	\draw [phantomlink] (#1) -- (#5);
	\draw [phantomlink] (#1) -- (#6);
	\draw [phantomlink] (#1) -- (#7);
	\draw [phantomlink] (#1) -- (#8);
	\draw [phantomlink] (#2) -- (#5);
	\draw [phantomlink] (#2) -- (#6);	
	\draw [phantomlink] (#2) -- (#7);
	\draw [phantomlink] (#2) -- (#8);	
	\draw [phantomlink] (#3) -- (#5);
	\draw [phantomlink] (#3) -- (#6);
	\draw [neglink] (#3) -- (#7);
	\draw [neglink] (#3) -- (#8);
	\draw [phantomlink] (#4) -- (#5);
	\draw [phantomlink] (#4) -- (#6);	
	\draw [poslink] (#4) -- (#7);
	\draw [poslink] (#4) -- (#8);
}

\newcommand\bs[3] %
{
	\draw [thin,fill=white,semitransparent] #1 +(-0.5*#2,-0.5*#3) rectangle +(0.5*#2,0.5*#3);
}

\newcommand\squeezedstate[4] %
{
	\draw [thin,gray] #1 +(-0.5*#2,0) -- +(0.5*#2,0) +(0,-0.5*#2) -- +(0,0.5*#2);
	\draw [thin] #1 circle (#3 and #4);
}

\newcommand{\op}[1]{\hat{#1}}
\newcommand{\opvec}[1]{\op{\vec{#1}}}
\newcommand{\opmat}[1]{\op{\mat{#1}}}
\newcommand{\id}{I}

\newcommand{\mat}[1]{\bm{\mathrm{#1}}}

\renewcommand{\vec}[1]{\bm{\mathrm{#1}}}

\newcommand{\blk}{\color{black}}

\begin{document}
\title{Passive interferometric symmetries of multimode Gaussian pure states}
\author{Natasha Gabay}
\email{natasha.gabay@sydney.edu.au}
\affiliation{School of Physics, The University of Sydney, Sydney, NSW 2006, Australia}
\author{Nicolas C. Menicucci}
\email{ncmenicucci@gmail.com}
\affiliation{School of Physics, The University of Sydney, Sydney, NSW 2006, Australia}
\affiliation{School of Science, RMIT University, Melbourne, VIC 3001, Australia}

\date{\today}
\begin{abstract}
As large-scale multimode Gaussian states begin to become accessible in the laboratory, their representation
and analysis become a useful topic of research in their own right. The graphical calculus for Gaussian pure
states provides powerful tools for their representation, while this work presents a useful tool for their analysis:
passive interferometric (i.e., number-conserving) symmetries. Here we show that these symmetries of multimode
Gaussian states simplify calculations in measurement-based quantum computing and provide constructive tools
for engineering large-scale harmonic systems with specific physical properties, and we provide a general
mathematical framework for deriving them. Such symmetries are generated by linear combinations of operators
expressed in the Schwinger representation of U(2), called nullifiers because the Gaussian state in question
is a zero eigenstate of them. This general framework is shown to have applications in the noise analysis of
continuous-various cluster states and is expected to have additional applications in future work with large-scale
multimode Gaussian states.
\blk 
\end{abstract}
\maketitle

\section{Introduction}\label{S:intro}
Optical quantum information science promises new and revolutionary technology~\cite{Weedbrook2012,Ralph2010}. Optical Gaussian states are states of light whose Wigner functions are Gaussian distributions over the real ($\op q$) and imaginary ($\op p$) parts of the complex quantised-mode (\textit{qumode}) amplitude, where $[\op q,\op p]=i\hbar$~\cite{Olivares2012}. In the context of quantum information, such states have the desirable feature of being produced, manipulated, and measured with experimental ease using squeezing, linear optics, and homodyne detection~\cite{Wang2008}. 

Many important aspects of quantum-information technologies have already been realised in Gaussian systems, including quantum teleportation~\cite{Furusawa1998}, entanglement swapping~\cite{Jia2004}, quantum dense coding~\cite{Li2002,Jing2003}, entanglement purification~\cite{Duan1999}, measurement-based QC~\cite{Gu2009}, cluster state preparation~\cite{Yukawa2008}, quantum error correction~\cite{Aoki2009}, and quantum algorithms~\cite{Zwierz2010}. These capabilities, combined with the relative experimental ease with which Gaussian states are processed and measured, encourages a serious consideration of Gaussian states for quantum information tasks. Indeed, Gaussian---or continuous-variable (CV)~\cite{Lloyd1999}---cluster states can be generated efficiently in a highly scalable fashion and serve as resources for measurement-based QC~\cite{Gu2009,Menicucci2006,Flammia2009,Zaidi2008,Yukawa2008,Yokoyama2013,Chen:2014jx}. 

A well-known mapping known as the \emph{Schwinger representation}~\cite{Schwinger} maps the state space of a pair of harmonic oscillators to the representation space of $\group {SU}(2)$~\cite{Chaturvedi2006} and thus to a single quantum spin. In this context, pairs of qumodes within an $n$-qumode system are known as \textit{Schwinger spins}.
The Schwinger mapping has been a beneficial tool when applied to many areas in quantum optics, including
the $\group {SU}(2)$ symmetry of a beamsplitter~\cite{Yurke1986,Campos1989}, demonstrating violations of Bell inequalities in macroscopic systems ~\cite{Reid2002,Su1991,Simon2003,Gerry2005,Evans2011}, the formulation of angular momentum coherent states~\cite{Atkins1971}, polarization squeezing and entanglement~\cite{Korolkova}, and detecting entanglement of non-Gaussian states~\cite{Nha2006}. These set-ups have applied the Schwinger representation to systems of a few qumodes, generally constructing a single Schwinger spin or a pair of them.
Considering the fruitful applications of applying the Schwinger representation to systems of a few qumodes, it is of interest to investigate many-qumode optical Gaussian states within the Schwinger picture. Such a mapping could also be of use regarding CV simulation of entangled spin systems. Recent work in this direction has posed similar motivation, proposing the need for a straightforward map between CV and Schwinger spin systems~\cite{Sridhar:2014jc}. A stepping stone to achieving such a map was suggested in~\cite{Sridhar:2014jc}, whereby the nullifiers for various multipartite Schwinger spin systems were derived and interpreted. 
In this work, we present necessary and sufficient conditions that must be satisfied by a Schwinger spin operator if it is to nullify a particular Gaussian pure state. Each nullifier generates an interferometric symmetry satisfied by the corresponding state. We present two applications of these results. First, we use these symmetries to identify a class of CV quantum gates whose members are all guaranteed to have the same noise properties when implemented in a particular (and useful) CV measurement-based setting. Second, we construct a class of quadratic parent Hamiltonians for any given Gaussian pure state and prove that, in this context, the interferometric symmetries of the ground state get promoted to conserved quantities for the full dynamics of this system. These results therefore are likely to have applications in quantum information, condensed matter, and quantum optics.\blk 
\section{Motivation and background}\label{backs}
\subsection{The Schwinger representation of \texorpdfstring{$\group U(2)$}{U(2)}}\label{Schwingerintro}

The Schwinger representation of $\group {SU}(2)$ \cite{Schwinger} maps a pair of harmonic oscillators to a single spin whose spin quantum numbers are determined by the sum and difference of quanta in each of the oscillators. It is a bosonic realization of a Lie algebra and is a mutiplicity-free direct sum of all the unitary irreducible representations of $\group {SU}(2)$ \cite{Chaturvedi2006}. Here we summarize the usual approach and generalize it to the Schwinger representation of $\group U(2)$.

In an optical setting, the harmonic oscillators used in the Schwinger representation of $\group {SU}(2)$ are qumodes that occupy a Fock space---i.e., an infinite dimensional Hilbert space endowed with the orthonormal basis $\{\ket{n}\}_{n=0}^{\infty}$ and operators $\op a$, $\op a^\dag$ satisfying $[\op a,\op a^\dag]=1$. The Schwinger spin operators are
\begin{align}
\op S^x_{1,2}&=\frac{1}{2}(\op a^\dag_{1}\op a_{2}+\op a^\dag_{2}\op a_{1})\,,\nonumber\\
\op S^y_{1,2}&=\frac{1}{2i}(\op a_{1}^\dag\op a_{2}-\op a_{2}^\dag\op a_{1})\,,\label{spinops}\\
\op S^z_{1,2}&=\frac{1}{2}(\op a^\dag_{1}\op a_{1}-\op a^\dag_{2}\op a_{2})\,,\nonumber
\end{align}
where the $1,2$ subscripts denote qumodes $1$ and $2$.

The appropriate $\algebra{su}(2)$ Lie algebra relations are satisfied:
\begin{align}
[\op S^k,\op S^l]=i\sum\limits_m \epsilon_{klm}\op S^m\,,
\end{align}
where $\epsilon_{klm}$ is the Levi-Civita symbol.
It is also important to introduce the $\op S^0$ operator,  
\begin{align}
\op S^0_{1,2}&=\frac{1}{2}(\op a^\dag_{1}\op a_{1}+\op a^\dag_{2}\op a_{2})\,,\label{s0}
\end{align}
which is related to $\op S^2=(\op S^x)^2+(\op S^y)^2+(\op S^z)^2$, the Casimir invariant of $\algebra{su}(2)$, as
\begin{align}
\op S^2=\op S^0(\op S^0+1)\,.
\end{align}
Therefore, the $\op S^0$ operator corresponds to the spin quantum number $s$,
\begin{align}
\op S^0\mapsto s \,,
\end{align}
while $\op S^z$ corresponds to the spin quantum number $m_s$,
\begin{align}
\op S^z\mapsto m_s \,.
\end{align} 

Combining the Schwinger $\op S^0$ operator with the three Schwinger spin operators, the Schwinger representation of $\algebra u(2)$ is defined by the generators:
\begin{align}
\op S^j:=\frac{1}{2}\opvec a^\herm \mat \sigma_j \opvec a\,, \qquad j \in \{0,x,y,z\}\,,\label{schwingergens}
\end{align}
where $\mat \sigma_j$ is the identity or a Pauli matrix, respectively, and $\opvec a=(\op a_1,\op a_2)^\tp$ and $\opvec a^\herm=(\opvec a^\dag)^\tp=(\op a_1^\dag,\op a_2^\dag)$ are vectors of operators. As the Schwinger operators (\ref{schwingergens}) form a representation of $\algebra u(2)$, any element $\op U\in \group U(2)$ can be written as 
\begin{align}
\op U(\vec{\theta})=\exp(-i\vec{\theta}\cdot\opvec S)\,, \label{U(2)}
\end{align}
where $\opvec S=(\op S^0,\op S^x,\op S^y,\op S^z)$, $\vec{\theta}=(\theta_0,\theta_x,\theta_y,\theta_z)$, and ${\theta_i \in\reals}$ for $i\in\{0,x,y,z\}$.

\subsection{Passive interferometry}\label{sec:passiveinter}
To describe local interferometric symmetries of optical Gaussian states, %
we first note that only Hamiltonians that are quadratic in the quadrature operators correspond to Gaussian unitary transformations---i.e., unitary operations that preserve the Gaussian form of the state.

There are two varieties of quadratic Hamiltonians: \textit{compact}, which correspond to transformations that preseve the total photon number, and \textit{non-compact}, which do not preserve the total photon number. In this work we are concerned with \emph{passive transformations}---i.e., those that are generated by compact Hamiltonians. These can be employed with interferometry (beamsplitters and phase shifters). Operators that generate passive transformations on $n$ qumodes are known to correspond to the group $\group U(n)$, a subgroup of $\group{Sp}(2n,\reals)$ \cite{Simon1994}, the latter being the group of all Gaussian unitary operations.  Considering a system of two qumodes, any passive transformation corresponds to elements of $\group U(2)$ and thus can be expressed in the form of Eq.~\eqref{U(2)}. Generalising to an $n$-qumode system, any $\group U(n)$ transformation can be decomposed into a sequence of $\group U(2)$ operations acting on two-dimensional subspaces of the underlying $n$-dimensional Hilbert space \cite{Reck1994}. Therefore, any $n$-qumode passive interferometer can interpreted as a sequence of two-qumode passive interferometers defined in Eq.~\eqref{U(2)} acting on two-qumode subspaces of the system.

The $\op S^x$ and $\op S^y$ generators correspond to \textit{two-mode mixing} operators \cite{Schumaker1986} that couple the qumodes together. Such operators can be modelled as different beamsplitter configurations \cite{Campos1989}. The $\op S^z$ generator---the \textit{relative-phase-shift operator}---corresponds to an equal and opposite phase shift being applied to each qumode. %
The fourth generator, $\op S^0$, imparts an equal phase shift on both qumodes. The unitary generated by this operator consists of two individual qumode phase-space rotation operators $\op R(\theta)=\exp(-i\theta\op a^\dag_j \op a_j)$, where $j$ denotes qumode~$j$. This operator acting on an eigenstate of the number operator simply results in an overall phase factor of $e^{-i\theta n}$. However, in general, such an operator has a nontrivial effect on a state.

Concerning some previous investigations of interferometry using group theory \cite{Campos1989,Yurke1986,Leonhardt1993}, it has been sufficient to only use a representation of $\group {SU}(2)$. This is because in those cases the effects of interest were either interference effects (that only depend on the relative phase difference between qumodes and are thus not altered by the $\op S^0$ generator) and effects that can be measured with photodetectors (photodetectors are insensitive to the transformations generated by $\op S^0$). We are nonetheless interested in the full Schwinger $\group U(2)$ representation, which describes the most general interferometric operation. This requires accounting for overall phase shifts on each 2-qumode subspace because on the full $n$-qumode space these will now become relative phase shifts, which are physically significant.
\subsection{The graphical calculus for Gaussian pure states}\label{introducegraphicalcalculus}
The nullifier formalism is well suited for analysis of Gaussian pure states including CV cluster states  \cite{Gu2009,Menicucci2006,Menicucci2011}. Any Gaussian pure state (including CV cluster states) can be defined uniquely (up to phase-space displacements) by an undirected graph with complex-valued edge weights, given by the matrix $\mat Z=\mat V+i\mat U$, where $\mat U=\mat U^\tp>0$ and $\mat V=\mat V^\tp$ \cite{Menicucci2011}. Every such $\mat Z$ also defines a unique Gaussian pure state. The position-space wavefunction for a Gaussian pure state on $n$ qumodes is related to $\mat Z$ by:
\begin{align}
\psi_{\mat Z}(\vec q)=\frac{(\det \mat U)^{1/4}}{\pi^{n/4}}\exp\left(\frac{i}{2}\vec q^\tp \mat Z \vec q\right)\,.
\end{align}

The nullifiers for a Gaussian pure state are also captured within this formalism. For an $n$-qumode Gaussian pure state $\ket{\psi_{\mat Z}}$, $n$ independent nullifiers can be written succinctly as
\begin{align}
(\opvec p-\mat Z \opvec q)\ket{\psi_{\mat Z}}=\mat 0\,,\label{null}
\end{align}
where $\opvec p=(\op p_{1},\dotsc,\op p_n)^\tp$ and $\opvec q=(\op q_{1},\dotsc,\op q_n)^\tp$ are column vectors of quadrature operators.

We now introduce the complex, symmetric matrix $\mat K$ that is related to $\mat Z$ by \cite{Menicucci2011}
\begin{align}
\mat K=(\mat \id+i\mat Z)(\mat \id-i\mat Z)^{-1}\label{K}
\end{align}
and thus also uniquely specifies the state. The symmetry of $\mat K$ is an immediate consequence of the symmetry of $\mat Z$. 

The requirements on $\mat U$ and $\mat V$ mentioned above induce the following restriction on~$\mat K$:
\begin{align}
\lVert\mat K\rVert < 1\,,
\end{align}
where $\lVert\mat K\rVert$ is the spectral norm (i.e., the largest singular value) of~$\mat K$. For proof of this, see Appendix \ref{A2}. 

The uniqueness of $\mat K$ enables a useful graphical representation of the Gaussian pure state that it defines, whereby the graphical depiction of the state is simple, yet precise, as it represents the wavefunction for the corresponding state \cite{Menicucci2011}. We illustrate this graphical formalism with an example using a two-mode squeezed (TMS) state and its corresponding $\mat K$ matrix. A TMS state is an entangled, bipartite Gaussian state of two qumodes that exhibits EPR-like correlations \cite{Drummond1988}. %
The $\mat K$ matrix for a TMS state is \cite{Menicucci2011}
\begin{align}
\mat K=
\begin{pmatrix}
0 & \tanh \alpha\\
\tanh \alpha & 0 
\end{pmatrix}\label{Ktms}\,,
\end{align} 
where $\alpha>0$ is an overall squeezing parameter. Figure \ref{fig:TMS} displays the graphical representation of $\mat K$, whereby qumodes are depicted by nodes and the edge weight corresponds to the off-diagonal terms in $\mat K$.

\begin{figure}
\begin{center}
\beginpgfgraphicnamed{graphics/pairtms}%
\begin{tikzpicture} [scale=1,label distance=-4pt, pin distance=3.2em]
	\footnotesize
	\def\hgraph{(-1,0)}
	\def\labelloc{(-0.5*1,-2.2*1)}
	\def\spingraph{(4,0)}
	\def\symbolloc{(1,0.90)}
	\draw [ultra thick, color=blue!60!black](-2,0) -- (0,0);
   \shade[shading=ball, ball color=black] (-2,0) circle (3pt);   
    \node at (-1,0.6) [anchor=north] {$\tanh \alpha$};
  \shade[shading=ball, ball color=black] (0,0) circle (3pt);
\end{tikzpicture}
\endpgfgraphicnamed
\caption{\label{fig:TMS}Precise graphical representation (described in Section \ref{introducegraphicalcalculus}) of a two-mode squeezed (TMS) state with adjacency matrix $\mat K$ (given in Equation (\ref{Ktms})), where $\alpha>0$ is an overall squeezing parameter.}
\end{center}
\end{figure}
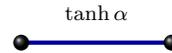

\begin{definition}
Any Gaussian pure state can be uniquely specified ---up to displacements and overall phase--- by the complex, symmetric matrix $ \mat K$, which we define to be the \emph{adjacency matrix} of the state.
\end{definition}
The nullifier relation of Equation (\ref{null}) can be written in terms of $\mat K$ and the qumode annihilation operators. First, we multiply Equation (\ref{null}) by $i$ on both sides, 
and then we use the relation between quadrature operators and annihilation operators, $\opvec a=\frac{1}{\sqrt{2}}(\opvec q+i\opvec p)$, to rewrite this expression as 
\begin{align}
(\opvec a-\opvec a^\dag)-i\mat Z(\opvec a+\opvec a^\dag)\ket{\psi_{\mat Z}}&=\mat 0\\
(\mat \id-i\mat Z)\opvec a-(\mat \id+i\mat Z)\opvec a^\dag\ket{\psi_{\mat Z}}&=\mat 0\,.
\end{align}
Multiplying both sides of this equation by $(\mat \id-i\mat Z)^{-1}$ results in the expression
\begin{align}
(\opvec a- \mat K \opvec a^\dag)\ket{\psi_{\mat Z}}=\mat 0\,.
\end{align}
For $\mat K$ defined through Equation (\ref{K}), the state defined by $\mat K$ is the same as the state defined by $\mat Z$ ($\ket{\phi_{\mat K}}=\ket{\psi_{\mat Z}}$) and thus, we can write the nullifier relation for a Gaussian pure state on $n$ qumodes in terms of $\mat K$: 
\begin{align}
(\opvec a-\mat K \opvec a^\dag)\ket{\phi_{\mat K}}=\mat 0\,.
\end{align}
This is written in terms of the column vectors of annihilation and creation operators $\opvec a=(\op a_1,\dotsc,\op a_n)^\tp$ and $\opvec a^\dag=(\op a^\dag_1,\dotsc,\op a^\dag_n)^\tp$. 

Now that we have provided the reader with the relevant graphical formalism for Gaussian pure states as well as the nullifier relations in terms of the adjacency matrix $\mat K$, we proceed to define Schwinger nullifiers for Gaussian pure states, the interferometric symmetries they generate, and several applications of this idea.
\section{Schwinger nullifiers for a Gaussian pure state}\label{schingernulls}

Here we derive the main technical results of the paper: necessary and sufficient conditions for a Schwinger operator to be a nullifier for a particular Gaussian pure state, as well as the symmetry that such an operator generates. 

\begin{definition}
Let $\ket{\phi_{\mat K}}$ be a Gaussian pure state on n qumodes. A \emph{Schwinger nullifier} for $\ket{\phi_{\mat K}}$ is a real linear combination of Schwinger $\group U(2)$ operators that is a nullifier for $\ket{\phi_{\mat K}}$. 
\end{definition}
\begin{lemma}\label{lemma1}
Any real linear combination of Schwinger $\group U(2)$ operators acting on $n$ qumodes can be written concisely as
\begin{align}
\opvec a^\herm \mat M\opvec a\,,\label{linearcombo}
\end{align}
where $\mat M$ is an $n\times n$ Hermitian matrix, the column vector $\opvec a=(\op a_1,\dotsc,\op a_n)^\tp$, and the row vector $\opvec a^\herm=(\op a^\dag_1,\dotsc,\op a^\dag_n)$. Furthermore, any expression of this form corresponds to some real linear combination of Schwinger $\group U(2)$ operators acting on $n$ qumodes.
\end{lemma}
\begin{proof}
Recall from Section \ref{sec:passiveinter} that any real linear combination of Schwinger $\group U(2)$ operators corresponds to a compact Hamiltonian. The most general compact Hamiltonian for a system of $n$ qumodes has the following form~\cite{Schumaker1986}: 
\begin{align}
\sum\limits_{i,j=1}^nM_{ij}\op a^\dag_i \op a_j \,,\label{schumaker}
\end{align}
where $M_{ij}=M^*_{ji}$.  
Using the row and column vectors of creation/annihilation operators defined above, %
we can rewrite expression (\ref{schumaker}) as 
\begin{align}
\sum\limits_{i,j=1}^nM_{ij}\op a^\dag_i \op a_j=\opvec a^\herm \mat M\opvec a\,,
\end{align}
where $\mat M$ is a Hermitian matrix with entries $M_{ij}$.

To prove the converse, consider a pair~$(r,s)$ of the $n$~qumodes. Let $\mat \sigma_j^{(r,s)}$ be the $n\times n$ Hermitian matrix whose entries are all zero except for $(r,r)$, $(r,s)$, $(s,r)$, and $(s,s)$, which contain the corresponding entries from the Pauli matrix $\mat \sigma_j$. Then, a Schwinger $\group U(2)$ operator acting on that pair can be written concisely as
\begin{align}
\op S^j_{r,s}:=\frac{1}{2}\opvec a^\herm \mat \sigma_j^{(r,s)} \opvec a\,, \qquad\qquad j\in\{0,x,y,z\}\,.
\end{align}
The collection of all $\mat \sigma_j^{(r,s)}$ for all index pairs $(r,s)$ and all $j \in \{ 0, x, y, z\}$ form an overcomplete basis for all $n \times n$ Hermitian operators. As such, we can write
\begin{align}
	\mat M = \sum_{(r,s)} \sum_j c_j^{(r,s)} \mat \sigma_j^{(r,s)}
\end{align}
for some real constants $c_j^{(r,s)}$. 
When $\mat M$ is sandwiched between the vectors of creation and annihilation operators, we have
\begin{align}
	\opvec a^\herm \mat M \opvec a = \sum_{(r,s)} \sum_j 2 c_j^{(r,s)} \op S^j_{r,s}\,,
\end{align}
which is explicitly a real linear combination of Schwinger $\group U(2)$ operators.
\end{proof}

\begin{theorem}\label{Main}
Let $\ket{\phi_{\mat K}}$ be a Gaussian pure state on $n$ qumodes. Let $\mat M$ be an $n\times n$ Hermitian matrix. Then, $\opvec a^\herm \mat M \opvec a$ is a Schwinger nullifier for $\ket{\phi_{\mat K}}$ if and only if ${\mat M \mat K=-(\mat M \mat K)^\tp}$. 
\end{theorem}
\begin{proof}
Starting with the known nullifier relation for a Gaussian pure state,  %
\begin{align}
(\opvec a-\mat K \opvec a^\dag)\ket{\phi_{\mat K}}=\mat 0\,,
\end{align}
we note that multiplying this expression on the left by $\opvec a^\herm \mat M$ results in the following (which is also a nullifier relation):
\begin{align}
(\opvec a^\herm \mat M \opvec a-\opvec a^\herm \mat M \mat K \opvec a^\dag)\ket{\phi_{\mat K}}=\mat 0\,.\label{twoterms}
\end{align}
By Lemma \ref{lemma1}, the first term in parentheses is a real linear combination of Schwinger $\group U(2)$ operators. %
The second term contains $\op a^\dag_i \op a^\dag_j$ and $\op a_i \op a_j$, terms which correspond to active transformations and are not linear combinations of Schwinger $\group U(2)$ operators. Thus, in order for expression (\ref{twoterms}) to be a Schwinger nullifier, the second term must vanish. 

Assume $\mat M \mat K=-(\mat M \mat K)^\tp$. We now show that the second term in (\ref{twoterms}) thus vanishes. %
We use the result that for any anti-symmetric, complex matrix $\mat A$, $\vec x^\tp \mat A \vec x=\mat 0$ for any vector $\vec x$, where $\vec x^\tp$ denotes the transpose of $\vec x$. Therefore, $\opvec a^\herm \mat M \mat K \opvec a^\dag=0$. Although we are dealing with vectors of operators in this case, we can still use this result because all the operators commute.

Conversely, assume $\mat M \mat K\neq-(\mat M \mat K)^\tp$. The second term in (\ref{twoterms}) never equals zero. To see this, 
decompose $\mat M \mat K$ into its symmetric and anti-symmetric components:
\begin{align}
\mat M \mat K=\frac{1}{2}[(\mat M \mat K+(\mat M \mat K)^\tp)+(\mat M \mat K-(\mat M \mat K)^\tp)]\,.
\end{align}
The second term in (\ref{twoterms}) can then be written as:
\begin{align}
\frac{1}{2}\opvec a^\herm[\mat M \mat K+(\mat M \mat K)^\tp]\opvec a^\dag+\frac{1}{2}\opvec a^\herm[\mat M \mat K-(\mat M \mat K)^\tp]\opvec a^\dag\,.
\end{align}
The second term in this expression vanishes (according to the same anti-symmetric argument used previously), resulting in
\begin{align}
\frac{1}{2}\opvec a^\herm[\mat M \mat K+(\mat M \mat K)^\tp]\opvec a^\dag\,.
\end{align}
This expression will never vanish. To see this, let $\mat M \mat K+(\mat M \mat K)^\tp=\mat A$ and note
\begin{align}
\opvec a^\herm \mat A\opvec a^\dag=\sum\limits_{i,j}\op a^\dag_i \op a^\dag_jA_{ij}\,.
\end{align}
There are two cases in which this expression could potentially vanish. The first is if $\mat A=\mat 0$, which is not the case as we are assuming $\mat M \mat K\neq-(\mat M \mat K)^\tp$. The second case would be if $\forall i, j$:
\begin{align}
A_{ij}\op a^\dag_i \op a^\dag_j + A_{ji}\op a^\dag_j \op a^\dag_i=0\,,
\end{align}
but since $A_{ij}=A_{ji}$, by assumption there exists some $i, j$ such that $A_{ij}\neq 0$, implying
\begin{align}
A_{ij}(\op a^\dag_i \op a^\dag_j + \op a^\dag_j \op a^\dag_i)\neq 0\,.
\end{align}
\end{proof}

\begin{theorem}\label{thm:toderiveM}
Let $\ket{\phi_{\mat K}}$ be a Gaussian pure state on n qumodes with bipartite adjacency matrix $\mat K$. Schwinger nullifiers can be derived for this state, simply using the singular value decomposition of $\mat K$ and an additional diagonal matrix.
\end{theorem}
\begin{proof}
Starting with bipartite $\mat K$:
\begin{align}
\mat K=\begin{pmatrix}
\mat 0 & \mat K_0\\
\mat K^\tp_0 & \mat 0
\end{pmatrix}\,,
\end{align}
we decompose $\mat K_0$ into its singular value decomposition:
\begin{align}
\mat K&=\begin{pmatrix}
\mat 0 & \mat U \mat \Sigma \mat V^\herm \\
\mat V^* \mat \Sigma \mat U^\tp & \mat 0
\end{pmatrix}\,,
\end{align}
where $\mat U$ and $\mat V$  are unitary $n\times n$ matrices, $\mat \Sigma$ is an $n\times n$ diagonal matrix with non-negative values, and $\mat V^\herm$ denotes the conjugate transpose of $\mat V$.\footnote{We reserve $^\dag$ for the Hermetian conjugate of an operator. See Ref.~\cite{Menicucci2011} for more details on this notation.}  
We can then write $\mat K$ as:
\begin{align}
\mat K=(\mat U\oplus\mat V^*)(\mat \sigma_x \otimes\mat \Sigma)(\mat U^\tp \oplus\mat V^\herm)\,,
\end{align}
where $\mat \sigma_x$ is the Pauli-$x$ matrix. 
If we define $\mat M$ to be of the form
\begin{align}
\mat M=(\mat U\oplus\mat V^*)(\mat \sigma_z \otimes\mat D)(\mat U^\herm \oplus\mat V^\tp)\,,
\end{align}
where $\mat D$ is any symmetric matrix that commutes with $\mat \Sigma$ and $\mat \sigma_z$ is the Pauli-$z$ matrix, the equation $\mat M \mat K=-(\mat M \mat K)^\tp$ holds. By Theorem \ref{Main}, %
this means that $\opvec a^\herm \mat M \opvec a$ is a Schwinger nullifier for $\ket{\phi_{\mat K}}$. %
\end{proof}

\begin{theorem}\label{nullsimplysymmetries}
Let $\ket{\phi_{\mat K}}$ be a Gaussian pure state on $n$ qumodes. Schwinger nullifiers for this state generate passive interferometric symmetries that apply to the state---i.e., passive interferometric operations that leave the state invariant. 
\end{theorem}

\begin{proof}
Let $\op N$ be a Schwinger nullifier for $\ket{\phi_{\mat K}}$. Therefore, it is a real linear combination of Schwinger $\group U(2)$ operators. Consider the unitary operator $\op U(\theta)=e^{-i\theta\op N}$, where $\theta\in\reals$. This unitary is generated by Schwinger $\group U(2)$ elements and therefore is a passive interferometric operation on two qumodes. As $\op N$ is a nullifier for $\ket{\phi_{\mat K}}$, $\op U(\theta)$ is a stabilizer for $\ket{\phi_{\mat K}}$ (i.e., $\op U(\theta)\ket{\phi_{\mat K}}=\ket{\phi_{\mat K}}$) and is therefore an operator that leaves the state invariant.
\end{proof}   

\section{Schwinger nullifiers for H-graph states}\label{sec:h-graphstates}
There exists a particular class of Gaussian pure states whose adjacency matrix $\mat K$ depends on a real, symmetric matrix $\mat G$:
\begin{align}
\mat K=\tanh(\alpha\mat G)\,,\label{eq:kintermsofG}
\end{align}
where $\alpha>0$ is an overall squeezing parameter. These states have been previously discussed in terms of the $\mat Z$ matrix formalism \cite{Menicucci2011}, where $\mat Z=ie^{-2\alpha\mat G}$. For proof of Equation (\ref{eq:kintermsofG}) in light of this formalism, see Appendix \ref{A1}. Such states can be generated via parametric down conversion in an optical parametric oscillator \cite{Pfister2004}, whereby the multi-qumode squeezing interaction is defined by the Hamiltonian
\begin{align}
\mathcal{H}=i\hbar\kappa\sum\limits_{m,n}G_{mn}(\op a^\dag_m \op a^\dag_n - \op a_m \op a_n)\,,
\end{align}
where $\kappa$ is a global coupling strength, and $\mat G$ is a matrix that specifies the multi-qumode interactions. 
Such states are called \textit{H-graph states} \cite{Flammia2009,Zaidi2008}, as they are defined in terms of a matrix (or \textit{graph}) $\mat G$ that determines the multimode-squeezing \textit{Hamiltonian}. H-graph states can also be generated via different optical methods. For instance, they can be generated by sending pairs of TMS states through a series of beamsplitters \cite{Menicucci2011a}.

\subsection{Schwinger nullifiers for H-graph states with bipartite, self-inverse adjacency matrices}\label{sec:methodtoderive}
 The class of H-graph states with self-inverse and bipartite $\mat G$ matrices are particularly useful for measurement-based quantum computation since it has been shown that these states are equivalent to CV cluster states up to certain phase shifts \cite{Flammia2009} and are able to generated in a highly scalable and compact fashion \cite{Pfister2004,Zaidi2008,Menicucci2007,Menicucci2011a,Wang:2014im}. H-graph states of this form that have a square-lattice structure are universal resources for measurement-based quantum computation \cite{Menicucci2011a,Wang:2014im,Flammia2009}.%
When $\mat G$ is self-inverse, $\mat K$ can be written as 
\begin{align}
\mat K=(\tanh \alpha)\mat G\,.
\end{align}
For proof of this, see Appendix \ref{A1}. 
When $\mat G$ is also bipartite,
\begin{align}
\mat K=(\tanh \alpha)\begin{pmatrix}
0 & \mat G_0\\
\mat G^\tp_0 & 0
\end{pmatrix}\,,\end{align}
which allows a straightforward derivation of $\mat M$ (as defined in Theorem \ref{thm:toderiveM}) as
\begin{align}
\mat M=(\mat U\oplus\mat V^*)(\mat \sigma_z \otimes\mat D)(\mat U^\herm \oplus\mat V^\tp)\,,
\end{align}
where $\mat U \mat \Sigma \mat V^\herm = \mat G_0$ is the singular value decomposition of $\mat G_0$, and $\mat D$ is any symmetric matrix that commutes with $\mat \Sigma$. 

This method can be used to derive local Schwinger nullifiers for any H-graph state with bipartite, self-inverse $\mat G$. The simplest example of such a graph is the TMS state shown in Fig.~\ref{fig:TMS} and for which $\mat G=\mat \sigma_x$ and therefore, $\mat M = \mat \sigma_z$ by the above construction. This means that $\op S_z$ nullifies a TMS state, as is well known~\cite{Korolkova}. A more substantial example is the \textit{dual-rail quantum wire}, which is a useful resource for CV measurement-based quantum computation~\cite{Menicucci2011a,Alexander:2014ew}. This is discussed next.

\section{The dual-rail quantum wire}\label{sec:qw}
As an example, we consider the dual-rail quantum wire (dual rail), which is a Gaussian cluster state with one-dimensional topology. Such a state can be generated by sending a collection of TMS states through beamsplitter operations \cite{Menicucci2011a}. %
Suitably weaving multiple dual rails together results in a structure that is equivalent to a 2D cluster state and is thus a universal resource for universal measurement-based quantum computation \cite{Flammia2009} (with fault tolerance possible if the squeezing is above a finite threshold~\cite{Menicucci:2014cx}).

We depict a dual rail on periodic boundaries in Figure \ref{fig:quantumwire}, which shows the precise graphical representation (discussed in Section \ref{introducegraphicalcalculus}) for the state, as well as additional labels that clarify how we re-interpret this state as entangled Schwinger spins. We interpret a dual rail on $2n$ qumodes as a collection of $n$ entangled Schwinger spins (labelled along the horizontal axis in Figure \ref{fig:quantumwire}), with qumodes paired vertically (labelled $a$ and $b$ along the vertical axis). A single qumode can thus be labelled according to which spin it is part of as well as where it lies on the vertical axis. For example, the first Schwinger spin in Figure \ref{fig:quantumwire} contains qumodes $1a$ and $1b$.

The method outlined in Section \ref{sec:methodtoderive} can be applied to derive local Schwinger nullifiers for the dual rail. In this section, we introduce another method that can be applied to such states, which is more intuitive. This method involves reinterpreting the dual rail (or any H-graph state with self-inverse and bipartite $\mat G$) in a different qumode decomposition.
\begin{figure}
\beginpgfgraphicnamed{graphics/cvwire}%
\begin{tikzpicture} [baseline=-2.5mm-0.5ex, x=1cm, y=5mm, inner sep=0pt, outer sep=0pt,scale=1.25,every node/.style={transform shape}]
	\def\lastn{5}
	\begin{scope}
	\clip ($ (.4,-1) - 0.55*(0,\micronodesize) $) rectangle ($ (\lastn+.6,0) + 0.55*(0,\micronodesize) $);
	\foreach \n in {0,...,\lastn}
	{
		\ifthenelse {\not\isodd{\n}}
			{
			\node (ia) [micro-even] at (\n,0) {};
			\node (ib) [micro-even] at (\n,-1) {};
			}
			{
			\node (ia) [micro-even] at (\n,0) {};
			\node (ib) [micro-even] at (\n,-1) {};
			}
		\node (oa) at (\n+1,0) {};
		\node (ob) at (\n+1,-1) {};
		\beamsplit {ia} {ib} {oa} {ob}
	}
		\ifthenelse {\isodd{\lastn}}
			{
			\node (ia) [micro-even] at (\lastn+1,0) {};
			\node (ib) [micro-even] at (\lastn+1,-1) {};
			}
			{
			\node (ia) [micro-even] at (\lastn+1,0) {};
			\node (ib) [micro-even] at (\lastn+1,-1) {};
			}
	\end{scope}
				
			\node at (-0.2,0.1) [anchor=north]{$a$};
			\node at (-0.2,-0.7) [anchor=north]{$b$};
			\node at (1,-2) [anchor=north]{$1$};
			\node at (2,-2) [anchor=north]{$2$};
			\node at (3,-2) [anchor=north]{$3$};
			\node at (4,-2) [anchor=north]{$4$};
			\node at (5,-2) [anchor=north]{$5$};
	\endpgfgraphicnamed
\end{tikzpicture}
\caption{\label{fig:quantumwire}The graph of the adjacency matrix ($\mat K$) for the dual-rail quantum wire on periodic boundary conditions. The weights of the links are $\pm\frac{1}{2}\tanh \alpha$, differentiated by the blue and yellow links, respectively. Schwinger spins are formed by pairing qumodes together vertically, where each Schwinger spin is labelled by integer values $1$--$5$ and contains two qumodes, labelled by $a$ and $b$.}
\end{figure}
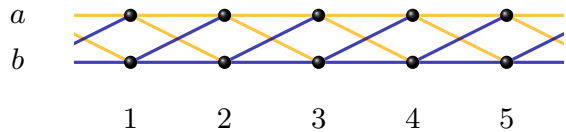

\subsection{Physical and distributed qumodes}\label{modechangetool}

It is possible to represent the dual rail in a different conceptual way, which is illustrated in Figure \ref{fig:distributedmodes}. This is an equivalent description of the state, whereby the labelling of the qumodes has changed from $(a,b)$ to~$(\pm)$. This represents a change of tensor-product structure---i.e.,~a different qumode decomposition of the two-qumode Hilbert space. Specifically, we have defined 
\begin{align}
\op a_{n\pm}=\frac{1}{\sqrt{2}}(\op a_{na}\pm\op a_{nb})\,.\label{modechange}
\end{align}
In effect, this change of tensor-product structure allows the beamsplitter interaction to be absorbed into the qumode redefinition and the system to be interpreted as a collection of TMS states in the $\pm$ qumodes even after the beamsplitter operation has been done \cite{Yokoyama2013,Alexander:2014ew}. We will call the $\pm$~qumodes the \textit{distributed qumodes} and the $(a,b)$~qumodes the \textit{physical qumodes} from now on. 
\begin{figure}
\begin{center}
\beginpgfgraphicnamed{graphics/tms}
\begin{tikzpicture} [baseline=-2.5mm-0.5ex, x=1cm, y=5mm, inner sep=0pt, outer sep=0pt,scale=1.25,every node/.style={transform shape}]
	\def\lastn{5}
	\begin{scope}
	\clip ($ (.4,-1) - 0.55*(0,\micronodesize) $) rectangle ($ (\lastn+.6,0) + 0.55*(0,\micronodesize) $);
	\foreach \n in {0,...,\lastn}
	{
		\ifthenelse {\not\isodd{\n}}
			{
			\node (ia) [micro-even] at (\n,0) {};
			\node (ib) [micro-even] at (\n,-1) {};
			}
			{
			\node (ia) [micro-even] at (\n,0) {};
			\node (ib) [micro-even] at (\n,-1) {};
			}
		\node (oa) at (\n+1,0) {};
		\node (ob) at (\n+1,-1) {};
		\tmslinky {ia} {ib} {oa} {ob};
	}
		\end{scope}
		\node at (-0.2,0.1) [anchor=north]{$+$};
		\node at (-0.2,-0.7) [anchor=north]{$-$};
			\node at (1,-2) [anchor=north]{$1$};
					\node at (2,-2) [anchor=north]{$2$};
					\node at (3,-2) [anchor=north]{$3$};
					\node at (4,-2) [anchor=north]{$4$};
					\node at (5,-2) [anchor=north]{$5$};

	\endpgfgraphicnamed
\end{tikzpicture}
\caption{\label{fig:distributedmodes}The graph of the adjacency matrix for the dual-rail quantum wire in the distributed qumodes, with the qumode-operator relabelling (in terms of $\pm$ along the vertical axis) as indicated in Equation (\ref{modechange}). The weights of the links are all $\tanh \alpha$ in this qumode decomposition.}
\end{center}
\end{figure}
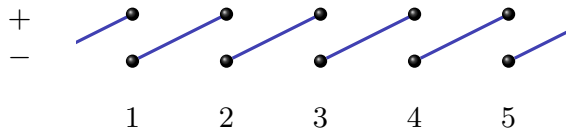

\subsection{Local nullifiers for the dual-rail quantum wire}
We now show how to derive local Schwinger nullifiers for the dual rail using its representation in distributed qumodes (Fig.~\ref{fig:distributedmodes}). %
The dual-rail wire in these modes is just a collection of TMS states in these modes. As discussed in Sec.~\ref{sec:methodtoderive}, a TMS state admits a nullifier of $\op S^z$ acting on the corresponding two (distributed) qumodes~\cite{Korolkova}. As such, each of these operators is also a nullifier for the full dual-rail wire. Expressing this in terms of the physical qumodes [using Eq.~\eqref{modechange}], this is equivalent to a sum of two Schwinger spin operators acting on neighbouring physical-qumode spins.

Concretely, one nullifier for the dual rail (see Figure \ref{fig:distributedmodes}) is $\op S^z_{2-,3+}$, which can be rewritten in the physical modes as 
\begin{align}
\op S^z_{2-,3+}&=\frac{1}{2}(\op a^\dag_{2-}\op a_{2-}-\op a^\dag_{3+}\op a_{3+})\nonumber\\
&=\frac{1}{2}(\op S^0_{2a,2b}-\op S^x_{2a,2b}-\op S^0_{3a,3b}-\op S^x_{3a,3b})\,.\label{end}
\end{align}

For a dual rail (on periodic boundaries) on $n$ Schwinger spins, this technique applied to every TMS state in the system results in $n$ local, independent Schwinger nullifiers:
\begin{align}
\op N_i= \op S^0_{ia,ib}-\op S^x_{ia,ib}-\op S^0_{(i+1)a,(i+1)b}-\op S^x_{(i+1)a,(i+1)b}\,,\label{localnuls}  
\end{align}
for all $i\in\Z_n$, and where the qumodes are labelled $a$ and $b$ and paired vertically into Schwinger spins that are labelled by $i$. Note that addition is modulo $n$.

The technique described in Section \ref{sec:methodtoderive}, which derives Schwinger nullifiers for any H-graph state with bipartite, self-inverse $\mat G$ could also have been used to derive the nullifiers in this section. However, the technique of switching qumode decomposition lends itself to a more natural and straightforward derivation and also provides us with a second method to derive such nullifiers in this case.

The nullifiers of Eq.~\eqref{localnuls} reveal information about the underlying Schwinger spin structure of the dual rail on periodic boundaries. For instance, they reveal that for any pair of adjacent Schwinger spins $i$ and $j$, the values obtained by measuring $\op S^0_{ia,ib}-\op S^0_{ja,jb}$ determine the total spin along the $x$-axis. The $\op S^0_{ia,ib}$ operator corresponds to the total photon number in the Schwinger spin $i$, and thus, the spin along the $x$-axis at any two adjacent spins on the ring of entangled Schwinger spins is directly determined by the photon-number difference between the spins.

It is worthwhile to note that these nullifiers do not solely generate rotations of the corresponding Schwinger spins, as these nullifiers include the $\op S^0$ operator, which is not a generator of the Schwinger representation of $\group {SU}(2)$. It is however, a generator of the Schwinger representation of $\group U(2)$ (as discussed in Section \ref{Schwingerintro}) and  corresponds to an equal phase shift being applied to both qumodes it acts on.

\subsection{Chain-like Schwinger spin nullifiers for the dual-rail quantum wire}
Any linear combination of the nullifiers given in Equation (\ref{localnuls}) is also a Schwinger nullifier. In this section, we show that by taking positive combinations of some of these local nullifiers, chain-like nullifiers emerge. These provide information about a segment of the ring of entangled Schwinger spins on the dual rail, where the length of the segment depends on the amount of nullifiers summed. We illustrate this by considering a dual-rail quantum wire on eight Schwinger spins with periodic boundaries, which has corresponding nullifiers given by Equation (\ref{localnuls}) with $n=8$. Summing nullifiers $\op N_1$ to $\op N_{3}$ results in a nullifier of the form:
\begin{align}\label{segment}
\op S^0_{1a,1b}-\op S^x_{1a,1b}-2\op S^x_{2a,2b}-2\op S^x_{3a,3b}-
\op S^0_{4a,4b}-\op S^x_{4a,4b}\,,
\end{align}
which we illustrate graphically in Figure \ref{fig:segment}. We call these chain-like nullifiers as it can be shown that by adding more nullifiers to the summed expression, the nullifier grows along the wire until all of the $n$ independent nullifiers have been summed, resulting in a global nullifier that is described in the next section.

These chain-like nullifiers reveal information about any segment of the dual rail. The nullifier that is illustrated in Figure \ref{fig:segment} reveals that the total spin along the $x$-axis for Schwinger spins $2$ and $3$ depends on the values obtained by measuring $\op S^0_{1a,1b}-\op S^x_{1a,1b}$ and $\op S^0_{4a,4b}-\op S^x_{4a,4b}$, which can be done with beamsplitters and photon counting. 

\begin{figure}
\begin{tikzpicture} [baseline=-2.5mm-0.5ex, x=1cm, y=5mm, inner sep=0pt, outer sep=0pt,scale=1.25,every node/.style={transform shape}]
	\def\lastn{5}
	\begin{scope}
	\clip ($ (.4,-1) - 0.55*(0,\micronodesize) $) rectangle ($ (\lastn+.6,0) + 0.55*(0,\micronodesize) $);
	\foreach \n in {0,...,\lastn}
	{
		\ifthenelse {\not\isodd{\n}}
			{
			\node (ia) [micro-even] at (\n,0) {};
			\node (ib) [micro-even] at (\n,-1) {};
			}
			{
			\node (ia) [micro-even] at (\n,0) {};
			\node (ib) [micro-even] at (\n,-1) {};
			}
		\node (oa) at (\n+1,0) {};
		\node (ob) at (\n+1,-1) {};
		\beamsplit {ia} {ib} {oa} {ob}
	}
		\ifthenelse {\isodd{\lastn}}
			{
			\node (ia) [micro-even] at (\lastn+1,0) {};
			\node (ib) [micro-even] at (\lastn+1,-1) {};
			}
			{
			\node (ia) [micro-even] at (\lastn+1,0) {};
			\node (ib) [micro-even] at (\lastn+1,-1) {};
			}
				\end{scope}
	\draw[left color=green, opacity=0.5](1,-0.5) ellipse (0.2 and 1.2);		
    \draw[left color=red, opacity=0.5](4,-0.5) ellipse (0.2 and 1.2);
     \draw[left color=blue, opacity=0.5](2,-0.5) ellipse (0.2 and 1.2);
       \draw[left color=blue, opacity=0.5](3,-0.5) ellipse (0.2 and 1.2);
	\node at (1,-2) [anchor=north]{$1$};
						\node at (-0.2,0.1) [anchor=north]{$a$};
							\node at (-0.2,-0.7) [anchor=north]{$b$};
	\node at (2,-2) [anchor=north]{$2$};
	\node at (3,-2) [anchor=north]{$3$};
	\node at (4,-2) [anchor=north]{$4$};
	\node at (5,-2) [anchor=north]{$5$};
\end{tikzpicture}
\caption{\label{fig:segment}A portion of the dual-rail quantum wire on periodic boundaries interpreted as a ring of Schwinger spins. The qumodes are paired vertically into Schwinger spins, labelled $1$--$5$. The coloured ellipses are an illustration of the chain-like nullifier that results from taking linear positive combinations of three of the $n$ independent nullifiers from Equation (\ref{localnuls}) (where the green and red ellipses correspond to $\op S^0_{1a,1b}-\op S^x_{1a,1b}$ and $-\op S^0_{4a,4b}-\op S^x_{4a,4b}$ respectively) and the $2\op S^x$ operator acts on spins $2$ and $3$, indicated by the blue ellipses.}
\end{figure}
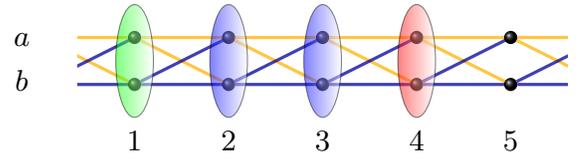
\subsection{Global Schwinger spin nullifiers for the dual-rail quantum wire}\label{sec:glob}
We just discussed how, by taking linear combinations of some of the local nullifiers presented in Equation (\ref{localnuls}), we can derive new nullifiers that act on segments of the dual rail. As we will show in this section, taking linear combinations of \textit{all} the independent nullifiers for a given system results in global nullifiers---i.e., nullifiers that act on every Schwinger spin in the system simultaneously. In terms of systems of entangled Schwinger spins, these global nullifiers reveal global properties of the system, such as what the overall spin is along a certain direction. As will become clear in this section, the differing global nullifiers also illustrate different spin structures---i.e., different potential qumode pairings.

For the dual rail on periodic boundaries comprising a total of $n$ Schwinger spins, by taking the sum of all nullifiers from the set given in Equation (\ref{localnuls}) we arrive at a global nullifier:
\begin{align}
\sum\limits_{i=1}^n\op S^x_{ia,ib}\,,\label{globx}
\end{align}
where $i$ denotes the Schwinger spin labelled $i$ along the horizontal direction and the $a$ and $b$ labels indicate the vertically paired qumodes. This is illustrated graphically in Figure \ref{fig:sxglobal}, which indicates the $\op S^x$ operator with blue ellipses. From this nullifier, we can conclude that the ring of entangled Schwinger spins (that is of variable length) always has a total spin along the $x$-axis equal to zero.

If $n$ is even, taking a sum of the differences of pairs of nullifiers from Equation (\ref{localnuls}) results in another global nullifier that acts on a different spin structure than the one that has been discussed so far. This new spin structure consists of qumodes that are paired \textit{horizontally} into spins. The following nullifier applies to the structure: 
\begin{align}
\sum_{i=1}^{n/2} \sum_{j\in\{a,b\}} \op S^z_{(2i-1)j,(2i)j}\,,\label{globz}
\end{align}
where $i$ denotes the Schwinger spin labelled $i$, which consists of horizontally paired qumodes. This is illustrated in Figure \ref{fig:szglobal}, where the purple ellipses correspond to the $\op S^z$ nullifiers acting the corresponding qumodes. Thus, a dual rail with an even number of Schwinger spins (made out of horizontally paired qumodes rather than vertically) is a ring of entangled Schwinger spins in which the total spin along the $z$-axis is equal to zero at all times.

\begin{figure}
\beginpgfgraphicnamed{graphics/cvwire}%
\begin{tikzpicture} [baseline=-2.5mm-0.5ex, x=1cm, y=5mm, inner sep=0pt, outer sep=0pt,scale=1.25,every node/.style={transform shape}]
	\def\lastn{4}
	\begin{scope}
		\clip ($ (.4,-1) - 0.55*(0,\micronodesize) $) rectangle ($ (\lastn+.6,0) + 0.55*(0,\micronodesize) $);
	\foreach \n in {0,...,\lastn}
	{
		\ifthenelse {\not\isodd{\n}}
			{
			\node (ia) [micro-even] at (\n,0) {};
			\node (ib) [micro-even] at (\n,-1) {};
			}
			{
			\node (ia) [micro-even] at (\n,0) {};
			\node (ib) [micro-even] at (\n,-1) {};
			}
		\node (oa) at (\n+1,0) {};
		\node (ob) at (\n+1,-1) {};
		\beamsplit {ia} {ib} {oa} {ob}
	}
		\ifthenelse {\isodd{\lastn}}
			{
			\node (ia) [micro-even] at (\lastn+1,0) {};
			\node (ib) [micro-even] at (\lastn+1,-1) {};
			}
			{
			\node (ia) [micro-even] at (\lastn+1,0) {};
			\node (ib) [micro-even] at (\lastn+1,-1) {};
			}
	\end{scope}
			\node at (1,-2) [anchor=north]{$1$};
			\node at (2,-2) [anchor=north]{$2$};
			\node at (3,-2) [anchor=north]{$3$};
			\node at (4,-2) [anchor=north]{$4$};
						\node at (-0.2,0.1) [anchor=north]{$a$};
						\node at (-0.2,-0.7) [anchor=north]{$b$};
		\draw[left color=blue, opacity=0.5](1,-0.5) ellipse (0.2 and 1.2);	
			\draw[left color=blue, opacity=0.5](2,-0.5) ellipse (0.2 and 1.2);	
				\draw[left color=blue, opacity=0.5](3,-0.5) ellipse (0.2 and 1.2);	
	    \draw[left color=blue, opacity=0.5](4,-0.5) ellipse (0.2 and 1.2);
	\endpgfgraphicnamed
\end{tikzpicture}
\caption{\label{fig:sxglobal}Schwinger spin structure for the dual-rail quantum wire on periodic boundary conditions, where the qumodes (labelled by $a$ and $b$) are paired together vertically into Schwinger spins, labelled 1--4. The blue ellipses are an illustration of the global nullifier for this state---i.e., a $\op S^x$ operator acting on every Schwinger spin of the system. }
\end{figure}
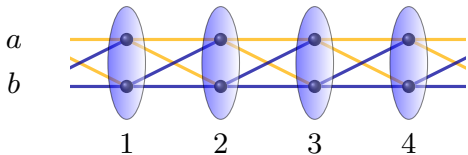
\section{Noise in CV measurement-based quantum computation}

As an application of these results, we show that the interferometric symmetries of the dual-rail quantum wire derived above can be used to deduce noise properties of measurement-based quantum computation using this resource state~\cite{Alexander:2014ew}.

We know from Theorem~\ref{nullsimplysymmetries} that Schwinger nullifiers generate passive interferometric symmetries. The local nullifier of Equation (\ref{end}), which acts on the dual rail on periodic boundary conditions $\ket{\psi}$, corresponds to the following local interferometric symmetry:
\begin{align}
e^{-i\theta(\op S^0_{2a,2b}-\op S^x_{2a,2b})}\ket{\psi}=e^{-i\theta(\op S^0_{3a,3b}+\op S^x_{3a,3b})}\ket{\psi}\,,\label{localsym}
\end{align} for all $\theta\in\reals$. In this section, we show how this symmetry explains a key result from recent work \cite{Alexander:2014ew}.

Measurement-based quantum computation on the dual rail in the physical qumodes is equivalent to sequential gate teleportation in the distributed qumodes ~\cite{Yokoyama2013}. As TMS states are not maximally entangled, any CV gate teleportation scheme never achieves perfect teleportation. Rather, there is always noise that distorts the information passing through the teleportation channel in a way that depends specifically on which gate is being teleported.

Interestingly, it was recently found \cite{Alexander:2014ew} that performing a rotation gate on the dual rail introduces the same amount of noise as performing the identity gate. %
This fact can be explained by the local symmetry of (\ref{localsym}), as we will now show. 

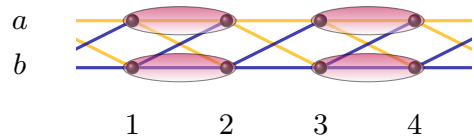
\begin{figure}
\beginpgfgraphicnamed{graphics/cvwire}%
\begin{tikzpicture} [baseline=-2.5mm-0.5ex, x=1cm, y=5mm, inner sep=0pt, outer sep=0pt,scale=1.25,every node/.style={transform shape}]
	\def\lastn{4}
	\begin{scope}
	\clip ($ (.4,-1) - 0.55*(0,\micronodesize) $) rectangle ($ (\lastn+.6,0) + 0.55*(0,\micronodesize) $);
	\foreach \n in {0,...,\lastn}
	{
		\ifthenelse {\not\isodd{\n}}
			{
			\node (ia) [micro-even] at (\n,0) {};
			\node (ib) [micro-even] at (\n,-1) {};
			}
			{
			\node (ia) [micro-even] at (\n,0) {};
			\node (ib) [micro-even] at (\n,-1) {};
			}
		\node (oa) at (\n+1,0) {};
		\node (ob) at (\n+1,-1) {};
		\beamsplit {ia} {ib} {oa} {ob}
	}
		\ifthenelse {\isodd{\lastn}}
			{
			\node (ia) [micro-even] at (\lastn+1,0) {};
			\node (ib) [micro-even] at (\lastn+1,-1) {};
			}
			{
			\node (ia) [micro-even] at (\lastn+1,0) {};
			\node (ib) [micro-even] at (\lastn+1,-1) {};
			}
	\end{scope}
	\draw[top color=purple, opacity=0.5](1.5,0) ellipse (0.6 and 0.3);		
	\draw[top color=purple, opacity=0.5](1.5,-1) ellipse (0.6 and 0.3);	
	\draw[top color=purple, opacity=0.5](3.5,0) ellipse (0.6 and 0.3);		
	\draw[top color=purple, opacity=0.5](3.5,-1) ellipse (0.6 and 0.3);	
		\node at (1,-2) [anchor=north]{$1$};
				\node at (2,-2) [anchor=north]{$2$};
				\node at (3,-2) [anchor=north]{$3$};
				\node at (4,-2) [anchor=north]{$4$};
							\node at (-0.2,0.1) [anchor=north]{$a$};
							\node at (-0.2,-0.7) [anchor=north]{$b$};
	\endpgfgraphicnamed
\end{tikzpicture}
\caption{\label{fig:szglobal}Alternate Schwinger spin structure for the dual-rail quantum wire on periodic boundary conditions (where the number of Schwinger spins is even). Note that the qumodes are now paired horizontally into spins. A global nullifier for this state is illustrated by the purple ellipses, which denote the $\op S^z$ operator acting on all of the Schwinger spins in the system.}
\end{figure}
\subsection{Two pictures of measurement-based quantum computing}

To discuss measurement-based quantum computing on the dual rail, we will vertically pair qumodes (which we label $a$ and $b$) into macronodes, which coincide with the spins shown in Fig.~\ref{fig:sxglobal}. 
Bare teleportation by one macronode on the dual-rail wire (with no ensuing phase shift) requires simply measuring $\op q_a$ and $\op p_b$~\cite{Yokoyama2013}. Preceding these same measurements by an equal phase shift of $\phi$ on each qumode still results in the information being teleported but now with a phase shift of~$2\phi$ subsequently being applied to the encoded information. This property was calculated explicitly in Refs.~\cite{Yokoyama2013,Alexander:2014ew}.

In fact, the symmetry of Eq.~\eqref{localsym} explains this fact intuitively. This symmetry allows us to replace the pre-measurement phase shifts on the measured qumodes with equivalent post-measurement phase shifts on the target qumodes, which will then be applied to the encoded quantum information \emph{after} it has been teleported. Both pictures are equivalent, but the symmetry-based explanation is more elegant and more powerful, and it can be generalized easily to other symmetries of other states.

The connection is not immediate, however, so we will show it explicitly in two parts. First, we will show~(a) that measuring $\op q_a$ and $\op p_b$ on the state represented by the left-hand side of Eq.~\eqref{localsym} can be effected by taking linear combinations of the outcomes of measurements of the rotated quadratures~$\op q_a^\theta$ and~$\op p_b^\theta$. These are defined by
\begin{align}
\begin{pmatrix}
	\op q^\theta	\\
	\op p^\theta
\end{pmatrix}
\coloneqq
\begin{pmatrix}
	\cos \theta	&	\sin \theta	\\
	-\sin \theta	&	\cos \theta
\end{pmatrix}
\begin{pmatrix}
	\op q	\\
	\op p
\end{pmatrix}
\,,
\end{align} 
which corresponds to the unitary Heisenberg action
\begin{align}
	\opvec x^\theta = \op R(\theta)^\dag \opvec x \op R(\theta)\,,
\end{align}
where $\op R(\theta)=\exp(-i\theta\op a^\dag \op a)$ is called a (positive) phase shift by~$\theta$. Second, we will show~(b) that the operation on the right-hand side of Eq.~\eqref{localsym} is is equivalent to a local phase shift on the teleported information.

The symplectic matrix corresponding to the Heisenberg action of $e^{-i2 \phi (\op S^0_{2a,2b}-\op S^x_{2a,2b})}$ is given by
\begin{align}
\mat S &=
\begin{pmatrix}
 c^2 & s^2 & c  s  & -c  s  \\
 s^2 & c^2 & -c  s  & c  s  \\
 -c  s  & c  s  & c^2 & s^2 \\
 c  s  & -c  s  & s^2 & c^2
\end{pmatrix}
\,,
\end{align}
where
$c = \cos \phi$ and $s = \sin \phi$. This means that multiplying the vector of quadrature operators
\begin{align}
\op{\vec x}=\begin{pmatrix}
\op q_{2a}\\
\op q_{2b}\\
\op p_{2a}\\
\op p_{2b}
\end{pmatrix} 
\end{align} by $\mat S$ gives the Heisenberg evolution of each operator and thus determines the rotated quadratures that are measured. Under this evolution ($\opvec x \mapsto \mat S \opvec x$),
\begin{align}
\op q_{2a} &\mapsto \op q_{2a}^\phi \cos  \phi - \op p_{2b}^\phi \sin \phi \nonumber\,,\\ %
\op p_{2b} &\mapsto \op q_{2a}^\phi \sin  \phi + \op p_{2b}^\phi \cos \phi \label{measu}\,.
\end{align}
Notice that these operators can be measured simply by measuring the (local) rotated quadratures~$\op q_{2a}^\theta$ and~$\op p_{2b}^\theta$ and then taking appropriate linear combinations of the results. This verifies~(a) above.

Next, we use Eq.~\eqref{modechange} to evaluate
\begin{align}
	\exp[-i2\phi(\op S^0_{3a,3b}+\op S^x_{3a,3b})] &= \exp(-i2\phi \op a^\dag_{3+}\op a_{3+}) \nonumber \\
	&= \op R_{3+}(2\phi)\,.
\end{align}
Since the quantum information in the dual-rail wire is carried in the symmetric~($+$) distributed qumode, this has the effect of phase shifting that information by~$2\phi$ after the teleportation. This verifies~(b) above.

The symmetry of (\ref{localsym}) thus implies that performing $e^{-i2\phi(\op S^0_{2a,2b}-\op S^x_{2a,2b})}$ is equivalent to applying $e^{-i2\phi(\op S^0_{3a,3b}+\op S^x_{3a,3b})}$ (illustrated in Figure  \ref{fig:localsymequiv}). The symmetry implies that applying a rotation gate can be described by two \textit{equivalent} scenarios: (a)~application of the unitary  $e^{-i2\phi(\op S^0_{2a,2b}-\op S^x_{2a,2b})}$ followed by measuring $\op q_{2a}$ and $\op p_{2b}$ and taking appropriate linear combinations of the results or~(b) measurements of $\op q_{2a}$ and $\op p_{2b}$ (and thus regular teleportation) followed by application of the unitary $ \op R_{3+}(2\phi)$ to the teleported information. The fact that this gate teleportation can be decomposed into regular teleportation followed by some unitary implies that this teleportation introduces the same amount of noise as performing ordinary teleportation. This was surprising when it was first shown in Ref.~\cite{Alexander:2014ew}, but now it can be seen to be a trivial result of the interferometric symmetry of the dual-rail wire shown in Fig.~\ref{fig:localsymequiv}.

\begin{figure}
\beginpgfgraphicnamed{graphics/cvwire}%
\begin{tikzpicture} [baseline=-2.5mm-0.5ex, x=1cm, y=5mm, inner sep=0pt, outer sep=0pt,scale=1.25,every node/.style={transform shape}]
	\def\lastn{2}
	\begin{scope}
	\clip ($ (.4,-1) - 0.55*(0,\micronodesize) $) rectangle ($ (\lastn+.6,0) + 0.55*(0,\micronodesize) $);
	\foreach \n in {0,...,\lastn}
	{
		\ifthenelse {\not\isodd{\n}}
			{
			\node (ia) [micro-even] at (\n,0) {};
			\node (ib) [micro-even] at (\n,-1) {};
			}
			{
			\node (ia) [micro-even] at (\n,0) {};
			\node (ib) [micro-even] at (\n,-1) {};
			}
		\node (oa) at (\n+1,0) {};
		\node (ob) at (\n+1,-1) {};
		\beamsplit {ia} {ib} {oa} {ob}
	}
		\ifthenelse {\isodd{\lastn}}
			{
			\node (ia) [micro-even] at (\lastn+1,0) {};
			\node (ib) [micro-even] at (\lastn+1,-1) {};
			}
			{
			\node (ia) [micro-even] at (\lastn+1,0) {};
			\node (ib) [micro-even] at (\lastn+1,-1) {};
			}
				\end{scope}

				\node at (-0.2,0.1) [anchor=north]{$a$};
				\node at (-0.2,-0.7) [anchor=north]{$b$};
	\node at (1,-2) [anchor=north]{$2$};
	\node at (2,-2) [anchor=north]{$3$};
	\node at (3,-0.5) {$=$};
		\def\lastn{6}
		\begin{scope}
		\clip ($ (4.4-1,-1) - 0.55*(0,\micronodesize) $) rectangle ($ (\lastn+.6-1,0) + 0.55*(0,\micronodesize) $);
		\foreach \n in {3,...,\lastn}
		{
			\ifthenelse {\not\isodd{\n}}
				{
				\node (ia) [micro-even] at (\n,0) {};
				\node (ib) [micro-even] at (\n,-1) {};
				}
				{
				\node (ia) [micro-even] at (\n,0) {};
				\node (ib) [micro-even] at (\n,-1) {};
				}
			\node (oa) at (\n+1,0) {};
			\node (ob) at (\n+1,-1) {};
			\beamsplit {ia} {ib} {oa} {ob}
		}
			\ifthenelse {\isodd{\lastn}}
				{
				\node (ia) [micro-even] at (\lastn+1,0) {};
				\node (ib) [micro-even] at (\lastn+1,-1) {};
				}
				{
				\node (ia) [micro-even] at (\lastn+1,0) {};
				\node (ib) [micro-even] at (\lastn+1,-1) {};
				}
	\end{scope}
		\draw[ line width = 0.5mm, left color=green, opacity=0.5](1,-0.5) ellipse (0.2 and 1.2);		

				\draw[line width = 0.5mm, left color=red, opacity=0.5](5,-0.5) ellipse (0.2 and 1.2);		
				\draw (1,-0.5) ellipse (0.2 and 1.2);
				\draw (5,-0.5) ellipse (0.2 and 1.2);
					\node at (-0.2,0.1) [anchor=north]{$a$};
					\node at (-0.2,-0.7) [anchor=north]{$b$};
		\node at (4,-2) [anchor=north]{$2$};
		\node at (5,-2) [anchor=north]{$3$};
	\endpgfgraphicnamed
\end{tikzpicture}
\caption{\label{fig:localsymequiv}Illustration of the local symmetry from (\ref{localsym}). Applying $e^{-i2\phi(\op S^0_{2a,2b}-\op S^x_{2a,2b})}$ (illustrated by the bold green ellipse) is equivalent to applying $e^{-i2\phi(\op S^0_{3a,3b}+\op S^x_{3a,3b})}$ (illustrated by the bold red ellipse).}
\end{figure}
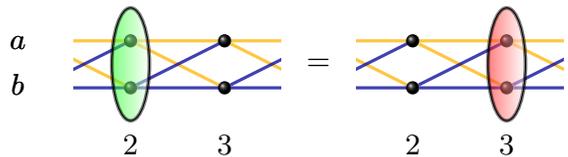

\
\subsection{Connection between CV cluster states and H-graph states}

A careful reading of Ref.~\cite{Alexander:2014ew} reveals that the quadratures prescribed for measurement in the scenario above are different from the ones presented here (even after accounting for the different conventions for the definition of positive phase shift). Specifically, Ref.~\cite{Alexander:2014ew} prescribes measuring $\op q_a^{\phi+\pi/4}$ and~$\op p_b^{\phi+\pi/4}$ to achieve the effect described above. Accounting for this difference is a different convention in talking about ``the dual-rail wire.''

Since the adjacency matrix~$\mat K$ for the dual-rail wire results from a bipartite, self-inverse H-graph, it is by now well established~\cite{Yokoyama2013,Menicucci2011a,Menicucci2011,Flammia2009,Menicucci2008} that this state can be transformed into a CV cluster state with the same basic form of the graph (for details, see references) simply by phase shifting all qumodes in one of the two graph bipartitions by $-\tfrac \pi 2$ (also known as a Fourier transform). In this case, it means performing this phase shift on all physical qumodes of all even (or all odd) macronodes of our version of the dual-rail wire~\cite{Menicucci2011a}.

The global symmetry illustrated in Fig.~\ref{fig:szglobal} shows, however, that by an argument analogous to Eq.~\eqref{localsym}, rather than phase shift half of the physical qumodes by $-\tfrac \pi 2$, one can choose instead to phase shift \emph{all} of the physical qumodes by $-\tfrac \pi 4$ to achieve the same effect.\footnote{If the graph has boundary conditions other than periodic, then one must modify this result for the ends of the chain, but the bulk of the chain remains unaffected by this, and thus we can ignore this detail without any loss of validity.} Therefore, in order to adapt the results above for use on the CV-cluster-state form of the dual-rail wire (used in Ref.~\cite{Alexander:2014ew}), we must first phase shift all modes by $+\tfrac \pi 4$ to convert to the H-graph-state form (used in the results above), measure $\op q_a^\phi$ and~$\op p_b^\phi$, and then phase shift all modes by $-\tfrac \pi 4$ to return to the CV-cluster-state form. All in all, this means measuring $\opvec x^\phi$ above is equivalent to measuring
\begin{align}
	 \op R\left (\frac \pi 4 \right)^\dag \opvec x^\phi \op R\left (\frac \pi 4 \right) = \opvec x^{\phi + \pi/4}
\end{align}
on the CV-cluster-state version. This agrees with the prescription given in Ref.~\cite{Alexander:2014ew}.

\section{Conservation laws for harmonic systems}\label{sec:groundstates}

As a second application of these results, in this section we interpret $\ket{\phi_{\mat K}}$ as the ground state of a two-body Hamiltonian $\op H(\mat K)$ and show that the Schwinger nullifiers commute with this Hamiltonian and thus reveal symmetries of the dynamics of the system. %
Recall that the following nullifier relation holds 
\begin{align}
(\opvec a-\mat K \opvec a^\dag)\ket{\phi_{\mat K}}=\vec 0\,.
\end{align}
We can multiply a nullifier by anything on the left and it remains a nullifier. We can therefore (as is done in \cite{Menicucci2011}) define annihilation and creation operators for the state defined by $\mat K$ as follows:
\begin{align}
\opvec a_{\mat K} &=\mat P^{1/2}(\opvec a-\mat K \opvec a^\dag)\,,\\
\opvec a^\herm_{\mat K} &=(\opvec a^\herm-\opvec a^\tp \mat K^*)\mat P^{1/2}\,,\label{ahk}
\end{align}
where $\mat P$ is some positive definite matrix. If we set $\mat P=(\mat \id-\mat K \mat K^*)^{-1}$, the ordinary commutation relations hold: $[\opvec{a}_{\mat K},\opvec a^\herm_{\mat K}]=\mat \id$. These operators can be used to derive the following (nonunique) Hamiltonian whose ground state is $\ket{\phi_{\mat K}}$:
\begin{align}\label{eq:HofK}
\op H(\mat K)=\opvec a^\herm_{\mat K}\opvec a_{\mat K}=
(\opvec a^\herm-\opvec a^\tp \mat K^*)\mat P(\opvec a-\mat K \opvec a^\dag)\,.
\end{align}
For any state $\ket{\phi_{\mat K}}$ and corresponding Schwinger nullifiers $\opvec a^\herm\mat M\opvec a$, the commutator $[\opvec a^\herm\mat M\opvec a,\op H(\mat K)]$ expands to
 
\begin{align}
[\opvec a^\herm \mat M\opvec a,\op H(\mat K)]
&=\opvec a^\herm \llbracket\mat M,\mat P \rrbracket\opvec a +\opvec a^\tp \llbracket\mat K^* \mat P \mat K, \mat M^* \rrbracket\opvec a^\dag \nonumber \\
&\quad -\opvec a^\herm(\mat M \mat P \mat K+\mat P \mat K \mat M^*)\opvec a^\dag\nonumber \\
&\quad +\opvec a^\tp(\mat M^* \mat K^* \mat P+\mat K^* \mat P \mat M)\opvec a\,,
\end{align} 
where $\llbracket \cdot, \cdot \rrbracket$ is the ordinary matrix commutator.\footnote{ Specifically, ${\llbracket \mat A, \mat B \rrbracket \coloneqq {\mat A\mat B - \mat B \mat A}}$. We need this new notation because we are following the conventions of Ref.~\cite{Menicucci2011}, which, among other things, defines the commutator for operator-valued matrices as
\begin{align}
	[\opmat A, \opmat B] \coloneqq \opmat A \opmat B - \Bigl(\opmat B^\tp \opmat A^\tp \Bigr)^\tp\,.	
\end{align}
Using this definition, $[\mat A, \mat B] = \mat 0$ for all c-number matrices even when $\llbracket \mat A, \mat B \rrbracket \neq \mat 0$. } Recall from Theorem~\ref{Main} that $\mat M \mat K = -(\mat M \mat K)^\tp$, $\mat K$ is symmetric, and $\mat M$ is Hermitian. Therefore, $\mat M \mat K = -\mat K\mat M^*$ and $\mat M^*\mat K^* = -\mat K^*\mat M$, which means that $\llbracket \mat M,\mat K\mat K^* \rrbracket=\mat 0$, and therefore $\llbracket\mat M,\mat P\rrbracket=\mat 0$. Using these relations to simplify the above shows that the commutator vanishes:
\begin{align}
	[\opvec a^\herm &\mat M\opvec a,\op H(\mat K)] = 0\,.
\end{align}
Therefore, for the quadratic Hamiltonian $\op H(\mat K)$ constructed from $\mat K$ by Eq.~\eqref{eq:HofK}, the Schwinger nullifiers of its ground state $\ket{\phi_{\mat K}}$ are conserved quantities for the dynamics of the system. \blk

As a simple example, consider the periodic dual-rail quantum wire on $2n$ qumodes with adjacency matrix $\mat K$  (introduced in Sec.~\ref{sec:qw}). Pairing qumodes vertically, this is equivalent to a ring of $n$ entangled Schwinger spins. Applying the results above to this state involves viewing the dual-rail quantum wire as the ground state of the two-body Hamiltonian $\op H(\mat K)=\opvec a^\herm_{\mat K}\opvec a_{\mat K}$.   As discussed in Section \ref{sec:glob}, there exists a global Schwinger nullifier which is a sum of $\op{S}^x$ operators acting on each spin in the system:
\begin{align}
\op S^x_{\text{tot}} \coloneqq \sum\limits_{i=1}^{n}\op{S}^x_{ia,ib}\,,\label{parent}
\end{align}
where $i$ denotes the Schwinger spin labelled $i$ along the horizontal direction and the $a$ and $b$ labels indicate the vertically paired qumodes. Thus, total spin along~$x$ is conserved in this system, as are each of the local 2-spin observables from Eq.~\eqref{localnuls}.  

In fact, we are not limited to quadratic Hamiltonians. One may add to $\op H(\mat K)$ other terms constructed solely from the nullifiers in order to produce a non-quadratic Hamiltonian with the same conservation laws but more complicated dynamics. (While $\ket{\phi_{\mat K}}$ will always remain an eigenstate of such a system, it may lose its privileged position as the ground state depending on the strength and signs of the additional terms.) We expect that this construction will open the doors to engineering bosonic systems that have particular desired properties in terms of their spin dynamics~\cite{Sridhar:2014jc}. 

\section{Conclusion}

We have presented a method for analyzing multimode Gaussian states in terms of their interferometric symmetries. These symmetries are generated by considering nullifiers for the state that are linear combinations of operators in the Schwinger representation of $\group U(2)$. In addition to producing a number of mathematical results, we applied this formalism to a straightforward analysis of noise in measurement-based quantum computation on the dual-rail quantum wire~\cite{Menicucci2011a,Yokoyama2013,Alexander:2014ew} and to the dynamics of harmonic and more general systems that have a Gaussian ground state.

As large-scale multimode Gaussian states become within reach in a compact laboratory setting~\cite{Alexander:2015wf,Chen:2014jx,Wang:2014im,Yokoyama:2013jp}, their representation~\cite{Menicucci2011} and analysis (this work) become all the more important. Furthermore, they are now becoming a testing ground for measurement-based quantum computing~\cite{Menicucci2006} and the creation and manipulation of topologically ordered quantum states~\cite{Demarie:2014jx}. As this research moves forward, the tools presented here---analysis in terms of number-conserving symmetries---will play an important role in the development of experimental protocols and theoretical understanding of Gaussian states and harmonic systems.

\acknowledgments
We are grateful for helpful discussions with Niranjan Sridhar, Olivier Pfister, Courtney Brell, and Rafael Alexander.  This work was supported by the Australian Research Council under grant No. DE120102204.
\appendix

\section{Proof of relationship between \texorpdfstring{$\mat K$}{K} and \texorpdfstring{$\mat G$}{G} and further simplification when \texorpdfstring{$\mat G$}{G} is self-inverse}\label{A1}
Noting the relationship between $\mat Z$ and the H-graph $\mat G$,
\begin{align}
\mat Z=ie^{-2\alpha\mat G}\,
\end{align}
and recalling the relationship between $\mat Z$ and $\mat K$, 
\begin{align}
\mat K=(\mat \id+i\mat Z)(\mat \id-i\mat Z)^{-1}\,, 
\end{align}
we can rewrite $\mat K$ as
\begin{align}
\mat K&=(\mat \id+i\mat Z)(\mat \id-i\mat Z)^{-1}\\
&=(\mat \id-e^{-2\alpha\mat G})(\mat \id+e^{-2\alpha\mat G})^{-1}\,.
\end{align}
Noting that $\mat G$ is diagonalizable and that 
\begin{align}
(1-e^{-2x})(1+e^{-2x})^{-1}=\tanh x\,,
\end{align}
we can write $\mat K$ as 
\begin{align}
\mat K&= %
\tanh(\alpha\mat G)\,.
\end{align}
Expanding this out in a Taylor series,
\begin{align}
\mat K &= %
\alpha\mat G-\frac 1 3(\alpha\mat G)^3 + \frac {2} {15} (\alpha\mat G)^5-\frac {17}{315} (\alpha\mat G)^7+\dotsm\,,
\end{align}
which converges for $|\alpha|<\frac{\pi}{2}$. 
If $\mat G$ is self-inverse, this simplifies to
\begin{align}
\mat K %
&= \left[ \alpha-\frac 1 3 \alpha^3 + \frac {2} {15} \alpha^5-\frac {17}{315} \alpha^7+\dotsm \right] \mat G
\nonumber\\
&\mathrel{=}(\tanh \alpha)\mat G\,.
\end{align}
\blk Even though the series expansion for $\tanh$ only holds when $|\alpha|<\frac{\pi}{2}$, since we are able to re-sum the series analytically, we can use analytic continuation to extend this result to all $\alpha\in\reals$.

\section{Derivation of relations for \texorpdfstring{$\mat K$}{K}}\label{A2}
The following relations between $\mat Z$ and $\mat K$ will be useful:
\begin{align}
	\frac 1 2 (\mat \id - i \mat Z)\phantom{^*} &= (\mat \id + \mat K)^{-1}\,, \\
	\frac 1 2 (\mat \id + i \mat Z^*) &= (\mat \id + \mat K^*)^{-1}\,, \\
	\frac 1 2 (\mat \id + i \mat Z)\phantom{^*} &= (\mat \id + \mat K^{-1})^{-1}\,, \\
	\frac 1 2 (\mat \id - i \mat Z^*) &= (\mat \id + \mat K^{-*})^{-1}\,.
\end{align}
To prove that ${\lVert\mat K\rVert<1}$, first we write $\mat U$ (defined in Section \ref{introducegraphicalcalculus}) in terms of $\mat K$ as follows:
\begin{align}
\mat U&=\frac{1}{2i}(\mat Z-\mat Z^*)\\
&=\frac{1}{2}(\mat \id-i\mat Z)+\frac{1}{2}(\mat \id+i\mat Z^*)-\mat \id\\
&=(\mat \id + \mat K)^{-1} + (\mat \id + \mat K^*)^{-1} - \mat \id \nonumber \\
&= (\mat \id + \mat K)^{-1} (\mat \id - \mat K\mat K^*) (\mat \id + \mat K^*)^{-1}\,.
\end{align}
Since $\mat K = \mat K^\tp$, and $(\mat \id + \mat K)$ is invertible, the condition that $\mat U > 0$ implies that $\mat \id - \mat K \mat K^* > 0$. This means that the eigenvalues of $\mat K\mat K^*$ are all less than~1. Since these eigenvalues are the squares of the singular values of~$\mat K$, the spectral norm condition follows.

\bibliographystyle{bibstyle}
\bibliography{bib,allrefs}

\begin{thebibliography}{10}

\bibitem{Weedbrook2012}
C. Weedbrook, S. Pirandola, R. Garc\'{\i}a-Patr\'{o}n, N.~J. Cerf, T.~C. Ralph,
  J.~H. Shapiro, and S. Lloyd, ``{Gaussian quantum information},'' Reviews of
  Modern Physics {\bf 84}, 621 (2012).

\bibitem{Ralph2010}
T.~C. Ralph and G.~J. Pryde, ``{Chapter 4: Optical Quantum Computation},''  in
  {\em Progress in Optics} (2010), Vol.~54, \ pp.\ 209--269.

\bibitem{Olivares2012}
S. Olivares, ``{Quantum optics in the phase space},'' The European Physical
  Journal Special Topics {\bf 203}, 3 (2012).

\bibitem{Wang2008}
X. Wang, T. Hiroshima, A. Tomita, and M. Hayashi, ``{Quantum information with
  Gaussian states},'' Physics reports {\bf 448}, 1 (2007).

\bibitem{Furusawa1998}
a. Furusawa, ``{Unconditional Quantum Teleportation},'' Science {\bf 282}, 706
  (1998).

\bibitem{Jia2004}
X. Jia, X. Su, Q. Pan, J. Gao, C. Xie, and K. Peng, ``{Experimental
  Demonstration of Unconditional Entanglement Swapping for Continuous
  Variables},'' Physical Review Letters {\bf 93}, 250503 (2004).

\bibitem{Li2002}
X. Li, Q. Pan, J. Jing, J. Zhang, C. Xie, and K. Peng, ``{Quantum Dense Coding
  Exploiting a Bright Einstein-Podolsky-Rosen Beam},'' Physical Review Letters
  {\bf 88}, 047904 (2002).

\bibitem{Jing2003}
J. Jing, J. Zhang, Y. Yan, F. Zhao, C. Xie, and K. Peng, ``{Experimental
  Demonstration of Tripartite Entanglement and Controlled Dense Coding for
  Continuous Variables},'' Physical Review Letters {\bf 90}, 167903 (2003).

\bibitem{Duan1999}
L.~M. Duan, G. Giedke, J.~I. Cirac, and P. Zoller, ``{Entanglement purification
  of gaussian continuous variable quantum states},'' Physical Review Letters
  {\bf 84}, 4002 (2000).

\bibitem{Gu2009}
M. Gu, C. Weedbrook, N.~C. Menicucci, T.~C. Ralph, and P. van Loock, ``{Quantum
  computing with continuous-variable clusters},'' Physical Review A {\bf 79},
  062318 (2009).

\bibitem{Yukawa2008}
M. Yukawa, R. Ukai, P. van Loock, and A. Furusawa, ``{Experimental generation
  of four-mode continuous-variable cluster states},'' Physical Review A {\bf
  78}, 012301 (2008).

\bibitem{Aoki2009}
T. Aoki, G. Takahashi, T. Kajiya, J.-i. Yoshikawa, S.~L. Braunstein, P.~V.
  Loock, and A. Furusawa, ``{Quantum error correction beyond qubits},'' Nature
  Physics {\bf 5}, 541 (2009).

\bibitem{Zwierz2010}
M. Zwierz, C. P\'{e}rez-Delgado, and P. Kok, ``{Unifying parameter estimation
  and the Deutsch-Jozsa algorithm for continuous variables},'' Physical Review
  A {\bf 82}, 042320 (2010).

\bibitem{Lloyd1999}
S. Lloyd and S.~L. Braunstein, ``{Quantum Computation over Continuous
  Variables},'' Physical Review Letters {\bf 82}, 1784 (1999).

\bibitem{Menicucci2006}
N.~C. Menicucci, P. van Loock, M. Gu, C. Weedbrook, T.~C. Ralph, and M.~A.
  Nielsen, ``{Universal Quantum Computation with Continuous-Variable Cluster
  States},'' Physical Review Letters {\bf 97}, 110501 (2006).

\bibitem{Flammia2009}
S.~T. Flammia, N.~C. Menicucci, and O. Pfister, ``{The optical frequency comb
  as a one-way quantum computer},'' Journal of Physics B: Atomic, Molecular and
  Optical Physics {\bf 42}, 114009 (2009).

\bibitem{Zaidi2008}
H. Zaidi, N.~C. Menicucci, S.~T. Flammia, R. Bloomer, M. Pysher, and O.
  Pfister, ``{Entangling the optical frequency comb: simultaneous generation of
  multiple 2x2 and 2x3 continuous-variable cluster states in a single optical
  parametric oscillator},'' Laser Phys. {\bf 18}, 659 (2008).

\bibitem{Yokoyama2013}
S. Yokoyama, R. Ukai, S.~C. Armstrong, C. Sornphiphatphong, T. Kaji, S. Suzuki,
  J.-i. Yoshikawa, H. Yonezawa, N.~C. Menicucci, and A. Furusawa,
  ``{Ultra-large-scale continuous-variable cluster states multiplexed in the
  time domain},'' Nature Photonics {\bf 7}, 982 (2013).

\bibitem{Chen:2014jx}
M. Chen, N.~C. Menicucci, and O. Pfister, ``{Experimental Realization of
  Multipartite Entanglement of 60 Modes of a Quantum Optical Frequency Comb},''
  Physical Review Letters {\bf 112}, 120505 (2014).

\bibitem{Schwinger}
J. Schwinger, ``{On Angular Momentum},''  in {\em Quantum theory of angular
  momentum}, 1st ed., {L. C. Biedenharn and H. V. Dam}, ed., (Academic Press,
  1965, 1965), \ pp.\ 229--279.

\bibitem{Chaturvedi2006}
S. Chaturvedi, G. Marmo, N. Mukunda, R. Simon, and A. Zampini, ``{The Schwinger
  representation of a group: concept and applications},'' Reviews in
  Mathematical Physics {\bf 18}, 887 (2006).

\bibitem{Yurke1986}
B. Yurke, S.~L. McCall, and J.~R. Klauder, ``{SU(2) and SU(1,1)
  interferometers},'' Physical Review A {\bf 33}, 4033 (1986).

\bibitem{Campos1989}
R.~A. Campos, B.~E. Saleh, and M.~C. Teich, ``{Quantum-mechanical lossless
  beamsplitter: SU(2) symmetry and photon statistics},'' Physical Review A {\bf
  40}, 1371 (1989).

\bibitem{Reid2002}
M. Reid, W. Munro, and F. {De Martini}, ``{Violation of multiparticle Bell
  inequalities for low- and high-flux parametric amplification using both
  vacuum and entangled input states},'' Physical Review A {\bf 66}, 033801
  (2002).

\bibitem{Su1991}
C. Su and K. W\'{o}dkiewicz, ``{Quantum versus stochastic or hidden-variable
  fluctuations in two-photon interference effects.},'' Physical Review A {\bf
  44}, 6097 (1991).

\bibitem{Simon2003}
C. Simon and D. Bouwmeester, ``{Theory of an Entanglement Laser},'' Physical
  Review Letters {\bf 91}, 053601 (2003).

\bibitem{Gerry2005}
C. Gerry and J. Albert, ``{Finite violations of a Bell inequality for high
  spin: An optical realization},'' Physical Review A {\bf 72}, 043822 (2005).

\bibitem{Evans2011}
R. Evans and O. Pfister, ``{On the experimental violation of mermin’s
  inequality with imperfect measurements},'' Quantum Information \& Computation
  {\bf 11}, 820 (2011).

\bibitem{Atkins1971}
P.~W. Atkins and J.~C. Dobson, ``{Angular Momentum Coherent States},''
  Proceedings of the Royal Society A: Mathematical, Physical and Engineering
  Sciences {\bf 321}, 321 (1971).

\bibitem{Korolkova}
N. Korolkova, G. Leuchs, R. Loudon, T.~C. Ralph, and C. Silberhorn,
  ``{Polarization squeezing and continuous- variable polarization
  entanglement},'' Physical Review A {\bf 65}, 052306 (2002).

\bibitem{Nha2006}
H. Nha and J. Kim, ``{Entanglement criteria via the uncertainty relations in
  su(2) and su(1,1) algebras: Detection of non-Gaussian entangled states},''
  Physical Review A {\bf 74}, 012317 (2006).

\bibitem{Sridhar:2014jc}
N. Sridhar and O. Pfister, ``{Generation of multipartite spin entanglement from
  multimode squeezed states},'' Phys. Rev. A {\bf 89}, 012310 (2014).

\bibitem{Simon1994}
R. Simon, N. Mukunda, and B. Dutta, ``{Quantum-noise matrix for multimode
  systems: U(n) invariance, squeezing, and normal forms},'' Physical Review A
  {\bf 49}, 1567 (1994).

\bibitem{Reck1994}
M. Reck, A. Zeilinger, H.~J. Bernstein, and P. Bertani, ``{Experimental
  realization of any discrete unitary operator},'' Physical Review Letters {\bf
  73}, 58 (1994).

\bibitem{Schumaker1986}
B. Schumaker, ``{Quantum mechanical pure states with Gaussian wave
  functions},'' Physics Reports {\bf 135}, 317 (1986).

\bibitem{Leonhardt1993}
U. Leonhardt, ``{Quantum statistics of a lossless beam splitter: SU(2) symmetry
  in phase space},'' Physical Review A {\bf 48}, 3265 (1993).

\bibitem{Menicucci2011}
N.~C. Menicucci, S.~T. Flammia, and P. van Loock, ``{Graphical calculus for
  Gaussian pure states},'' Physical Review A {\bf 83}, 042335 (2011).

\bibitem{Drummond1988}
M.~D. Reid and P.~D. Drummond, ``{Quantum Correlations of Phase in
  Nondegenerate Parametric Oscillation},'' Physical Review Letters {\bf 60},
  2731 (1988).

\bibitem{Pfister2004}
O. Pfister, S. Feng, G. Jennings, R. Pooser, and D. Xie, ``{Multipartite
  continuous-variable entanglement from concurrent nonlinearities},'' Physical
  Review A {\bf 70}, 020302 (2004).

\bibitem{Menicucci2011a}
N.~C. Menicucci, ``{Temporal-mode continuous-variable cluster states using
  linear optics},'' Physical Review A {\bf 83}, 062314 (2011).

\bibitem{Menicucci2007}
N.~C. Menicucci, S.~T. Flammia, H. Zaidi, and O. Pfister, ``{Ultracompact
  generation of continuous-variable cluster states},'' Physical Review A {\bf
  76}, 010302 (2007).

\bibitem{Wang:2014im}
P. Wang, M. Chen, N.~C. Menicucci, and O. Pfister, ``{Weaving quantum optical
  frequency combs into continuous-variable hypercubic cluster states},'' Phys.
  Rev. A {\bf 90}, 032325 (2014).

\bibitem{Alexander:2014ew}
R.~N. Alexander, S.~C. Armstrong, R. Ukai, and N.~C. Menicucci, ``{Noise
  analysis of single-mode Gaussian operations using continuous-variable cluster
  states},'' Phys. Rev. A {\bf 90}, 062324 (2014).

\bibitem{Menicucci:2014cx}
N.~C. Menicucci, ``{Fault-Tolerant Measurement-Based Quantum Computing with
  Continuous-Variable Cluster States},'' Physical Review Letters {\bf 112},
  120504 (2014).

\bibitem{Menicucci2008}
N. Menicucci, S. Flammia, and O. Pfister, ``{One-Way Quantum Computing in the
  Optical Frequency Comb},'' Physical Review Letters {\bf 101}, 130501 (2008).

\bibitem{Alexander:2015wf}
R.~N. Alexander, P. Wang, N. Sridhar, M. Chen, O. Pfister, and N.~C. Menicucci,
  ``{One-way quantum computing with arbitrarily large time-frequency
  continuous-variable cluster states from a single optical parametric
  oscillator},'' arXiv:1509.00484v1 [quant-ph]  (2015).

\bibitem{Yokoyama:2013jp}
S. Yokoyama, R. Ukai, S.~C. Armstrong, C. sornphiphatphong, T. kaji, S. Suzuki,
  J.-i. Yoshikawa, H. Yonezawa, N.~C. Menicucci, and A. Furusawa,
  ``{Ultra-large-scale continuous-variable cluster states multiplexed in the
  time domain},'' Nature Photonics {\bf 7}, 982 (2013).

\bibitem{Demarie:2014jx}
T.~F. Demarie, T. Linjordet, N.~C. Menicucci, and G.~K. Brennen, ``{Detecting
  topological entanglement entropy in a lattice of quantum harmonic
  oscillators},'' New J. Phys. {\bf 16}, 085011 (2014).

\end{thebibliography}

\end{document}